\documentclass[sn-mathphys]{sn-jnl}

\usepackage{graphicx}
\usepackage{amsmath}
\usepackage{amssymb}
\usepackage{bm}
\usepackage{nicefrac}
\usepackage{booktabs}
\usepackage[tableposition=top]{caption}
\usepackage{underscore}
\usepackage[utf8]{inputenc}
\usepackage[english]{babel}
\usepackage{array}
\usepackage{colortbl}
\usepackage{booktabs}
\usepackage{eurosym}
\usepackage{amsmath}
\usepackage{xcolor}

\usepackage{array}
\usepackage{multirow}
\usepackage{threeparttable}
\usepackage{makecell}
\usepackage{siunitx}
\usepackage{stfloats}
\usepackage{graphicx}
\usepackage{subfigure}
\usepackage{hyperref}
\usepackage{textcomp}
\usepackage{array}
\usepackage{epstopdf}
\usepackage{framed}
\usepackage{cleveref}
\usepackage[normalem]{ulem}
\usepackage{xcolor}
\usepackage{amsfonts}
\usepackage{booktabs}
\usepackage{siunitx}
\usepackage{algorithm}

\jyear{2023}%

\newcommand{\rea}{\mathbb{R}}
\theoremstyle{thmstylethree}
\newtheorem{problem}{Problem}
\newtheorem{theorem}{Theorem}[section]

\newtheorem{lemma}[theorem]{Lemma}

\newtheorem{definition}[theorem]{Definition}
\newtheorem{remark}[theorem]{Remark}

\newcommand{\PreserveBackslash}[1]{\let\temp=\\#1\let\\=\temp}
\newcolumntype{C}[1]{>{\PreserveBackslash\centering}p{#1}}
\newcolumntype{R}[1]{>{\PreserveBackslash\raggedleft}p{#1}}
\newcolumntype{L}[1]{>{\PreserveBackslash\raggedright}p{#1}}

\newenvironment{fminipage}%
{\begin{Sbox}\begin{minipage}}%
		{\end{minipage}\end{Sbox}\fbox{\TheSbox}}

\def\norm#1{\left\|  #1 \right\|}

\newcommand\Otil{\widetilde{O}}

\def\calG{\mathcal{G}}

\def\calN{\mathcal{N}}

\newcommand\LL{\bm{\mathit{L}}}

\def\defeq{\stackrel{\mathrm{def}}{=}}

\def\aa{\pmb{\mathit{a}}}

\newcommand\yy{\boldsymbol{\mathit{y}}}

\newcommand\xx{\boldsymbol{\mathit{x}}}
\newcommand\aaa{\boldsymbol{\mathit{a}}}

\newcommand\bb{\boldsymbol{\mathit{b}}}
\newcommand\cc{\boldsymbol{\mathit{c}}}

\newcommand\ee{\boldsymbol{\mathit{e}}}

\newcommand\uu{\boldsymbol{\mathit{u}}}

\renewcommand\AA{\boldsymbol{\mathit{A}}}
\newcommand\BB{\boldsymbol{\mathit{B}}}

\newcommand\CC{\boldsymbol{\mathit{C}}}

\newcommand\DD{\boldsymbol{\mathit{D}}}

\newcommand\EE{\boldsymbol{\mathit{E}}}

\newcommand\MM{\boldsymbol{\mathit{M}}}
\newcommand\TT{\boldsymbol{\mathit{T}}}

\newcommand\vv{\boldsymbol{\mathit{v}}}
\newcommand{\SDDMSolver}{\textsc{Solve}}

\usepackage{tabularx}
\usepackage{stfloats}

\begin{document}

\title[Maximizing the Smallest Eigenvalue of Grounded Laplacian Matrix]{Maximizing the Smallest Eigenvalue of Grounded Laplacian Matrix} 

\author[1]{\fnm{Xiaotian} \sur{Zhou}}\email{22110240080@m.fudan.edu.cn}
\author[1]{\fnm{Run} \sur{Wang}}\email{runwang18@fudan.edu.cn}
\author*[2]{\fnm{Wei} \sur{Li}}\email{fd\_liwei@fudan.edu.cn}
\author*[1,3,4]{\fnm{Zhongzhi} \sur{Zhang}}\email{zhangzz@fudan.edu.cn}

\affil[1]{\orgdiv{Shanghai Key Laboratory of Intelligent Information Processing, School of Computer Science, Fudan University}, \orgaddress{\city{Shanghai}, \postcode{200433}, \country{China}}}

\affil[2]{\orgdiv{Academy for Engineering and Technology, Fudan University}, \orgaddress{\city{Shanghai}, \postcode{200433}, \country{China}}}

\affil[3]{\orgdiv{Shanghai Engineering Research Institute of Blockchains, Fudan University}, \orgaddress{\city{Shanghai}, \postcode{200433}, \country{China}}}

\affil[4]{\orgdiv{Research Institute of Intelligent Complex Systems, Fudan University}, \orgaddress{\city{Shanghai}, \postcode{200433}, \country{China}}}



\abstract{For a connected graph $\calG=(V,E)$ with $n$ nodes, $m$ edges, and  Laplacian matrix $\LL$, a grounded Laplacian matrix $\LL(S)$ of $\calG$ is a $(n-k) \times (n-k)$ principal submatrix of $\LL$, obtained from $\LL$ by deleting $k$ rows and columns corresponding to $k$ selected nodes forming a set $S\subseteq V$. The smallest eigenvalue $\lambda(S)$ of $\LL(S)$ plays a pivotal role in various  dynamics defined on $\calG$. For example, $\lambda(S)$ characterizes the  convergence rate of leader-follower consensus, as well as the effectiveness of a pinning scheme for the pinning control problem, with  larger $\lambda(S)$ corresponding to smaller convergence time or better effectiveness of a pinning scheme.   In this paper, we focus on the problem of optimally selecting a subset $S$ of fixed $k \ll n$ nodes, in order to maximize the smallest eigenvalue $\lambda(S)$ of the grounded Laplacian matrix $\LL(S)$. We show that this optimization problem is NP-hard and that the objective function is non-submodular but monotone. Due to the difficulty to obtain the optimal solution, we first propose a na\"{\i}ve heuristic algorithm selecting one optimal node at each time for $k$ iterations. Then we propose a fast heuristic scalable algorithm to approximately solve this problem,  using derivative matrix, matrix perturbations, and Laplacian solvers as tools. Our na\"{\i}ve heuristic algorithm takes $\tilde{O}(knm)$ time, while the fast greedy heuristic has a nearly linear time complexity of $\tilde{O}(km)$, where $\tilde{O}(\cdot)$ notation suppresses the ${\rm poly} (\log n)$ factors. We also conduct numerous experiments on different networks sized up to one million nodes, demonstrating the superiority of our algorithm in terms of efficiency and effectiveness compared to baseline methods.}

\keywords{Grounded Laplacian, spectral property, combinatorial optimization, graph mining, linear algorithm, matrix perturbation, partial derivative, pinning control, convergence speed}

\maketitle

\section{Introduction}

%
The spectrum of different matrices associated with a graph provides rich structural and dynamical information of the graph, and it has found numerous applications in different areas~\cite{Ne03}. For example, the reciprocal of the largest eigenvalue of the adjacency matrix is approximately equal to the thresholds for the susceptible-infectious-susceptible epidemic dynamics~\cite{WaChWaFa03,ChWaWaLeFa08,VaOmKo08} and bond percolation~\cite{BoBoChRi10} on a graph. Concerning the Laplacian matrix, its smallest and largest nonzero eigenvalues are closely related to the time of convergence and delay robustness of the consensus problem~\cite{OlMu04}; all the nonzero eigenvalues determine the number of spanning trees~\cite{LiPaYiZh20} and the sum of resistance distances over all node pairs~\cite{KlRa93,LiZh18}, with the latter encoding the performance of different dynamical processes, such as the total hitting times of random walks~\cite{Te91,ChRaRuSm89,ShZh19}, robustness to noise in consensus problem~\cite{BaJoMiPa12,VeBoPa15,QiZhYiLi19,YiZhPa20}, and application to system control~\cite{KaVePa17,ChZhZhYuLi21,ChGaZhZh21,GoPeBiVo21,ZhZhLuLi23}.

In addition to the adjacency matrix and Laplacian matrix, the spectrum of the grounded Laplacian matrix also plays important role in different systems~\cite{RaJiMeEg09,PaBa10, PiShFiSu18, LiXuLuChZe21}. For a connected graph $\calG=(V,E)$ with Laplacian matrix $\LL$, the grounded Laplacian matrix $\LL(S)$ induced by a subset $S\subseteq V$ of $k$ nodes (called grounded nodes) is a principal submatrix of $\LL$, obtained from $\LL$ by removing $k$ rows and columns corresponding to the $k$ nodes in $S$~\cite{Mi93}. The sum of the reciprocal of eigenvalues of the grounded Laplacian matrix $\LL(S)$ captures the robustness performance of leader-follower systems with follower nodes subject to noise~\cite{PaBa10}. The magnitude of the smallest eigenvalue $\lambda(S)$ of matrix $\LL(S)$ determines the convergence rate of a leader-follower networked dynamical system~\cite{RaJiMeEg09}, as well as the effectiveness of pinning scheme of pinning control of complex dynamical networks~\cite{LiXuLuChZe21}, with large $\lambda(S)$ corresponding to fast convergence speed and good pinning control performance.

The smallest eigenvalue $\lambda(S)$ of $\LL(S)$ is an increasing function of set $S$ of nodes, corresponding to different physical meanings for different dynamical processes  or systems. In leader-follower consensus systems, $S$ is associated with the set of leader nodes~\cite{RaJiMeEg09}, while for the pinning control problem, $S$ corresponds to the set of controlled nodes~\cite{LiXuLuChZe21}. Since $\lambda(S)$ characterizes the performance of associated dynamical systems and is determined by the set $S$ of grounded nodes, a spontaneous question arises as what nodes should be chosen as grounded nodes such that $\lambda(S)$ is maximized. This is the theme of the current work. Concretely, we focus on the following  eigenvalue optimization problem: Given a graph $\calG=(V,E)$ with $\vert V\vert =n$ nodes and $\vert E\vert =m$ edges, and an integer $0<k  \ll n$, how to select a subset $S\subset V$ of $\vert S\vert = k$ nodes such that the smallest eigenvalue $\lambda(S)$ of the grounded Laplacian matrix $\LL(S)$ induced by the grounded nodes in $S$ is maximized.

The main contribution and work of this paper are summarized as follows.
\begin{itemize}
\item We show that the combinatorial optimization problem is NP-hard and that the optimization function is not submodular, although it is monotone.

\item We define a centrality of a group of nodes, called grounded node group centrality, which is the smallest eigenvalue of the grounded Laplacian matrix induced by the group of nodes. 

\item We propose an efficient, scalable algorithm to find $k$ nodes with the largest grounded node group centrality score. Our algorithm is a greedy heuristic one, which is established based on the tools of the derivative matrix, matrix perturbations,
and Laplacian solvers, and has a nearly linear time complexity of $\tilde{O}(km)$, with $\tilde{O}(\cdot)$ notation suppressing the ${\rm poly} (\log n)$ factors.

\item We evaluate our algorithm by performing extensive experiments on many real-world graphs up to one million nodes, which show that the proposed algorithm is effective and efficient, outperforming other baseline choices of selecting nodes.
\end{itemize}

\section{Related Work}
\label{section:relate}

In this section, we review some current  work related to the present one, including applications and properties of the grounded Laplacian matrix $\LL(S)$, node group centrality measures,  grounded node selection problem for maximizing the smallest eigenvalue $\lambda(S)$ of  $\LL(S)$ and its associated algorithms.

The grounded Laplacian matrix was first defined in~\cite{Mi93}. For a connected graph $\calG=(V, E)$, the grounded Laplacian matrix $\LL(S)$ induced by a subset $S\subseteq V$ of $k$ grounded nodes is a principal submatrix of its Laplacian matrix $\LL$, which is obtained from $\LL$ by deleting $k$ rows and columns associated with  the $k$ grounded nodes in $S$. Matrix $\LL(S)$ arises naturally in various practical scenarios, including leader-follower multi-agent systems~\cite{BaHe06,PaBa10}, pinning control of complex networks~\cite{WaCh02,LiWaCh04}, vehicle platooning~\cite{BaJoMiPa12,HeMaHuSe15}, and power systems~\cite{TeBaGa15}. It is now established that the spectrum of the grounded Laplacian matrix plays a fundamental role in characterizing the performance of these networked systems. For example, the smallest eigenvalue $\lambda(S)$ of $\LL(S)$ characterizes the convergence rate of leader-follower systems~\cite{RaJiMeEg09} and the effectiveness of pinning control scheme~\cite{LiXuLuChZe21}.

Because of the relevance, in recent years, there has been vast literature dedicated to analyzing the spectral properties of the grounded Laplacian matrix, particularly the smallest eigenvalue $\lambda(S)$. It was shown that~\cite{PaBa10}, eigenvalue $\lambda(S)$ is an increasing function of the grounded nodes set $S$, that is, $\lambda(S)$ does not decrease when new nodes are added to $S$. In~\cite{PiSu14,PiSu16, PiShSu15}, upper and lower bounds on the eigenvalue $\lambda(S)$ were provided, in terms of the sum of the weights of the edges between the grounded nodes in $S$ and the eigenvector corresponding to $\lambda(S)$. In~\cite{PiShFiSu18}, graph–theoretic bounds were given for the eigenvalue $\lambda(S)$. While in~\cite{LiXuLuChZe21}, properties of  $\lambda(S)$ were analyzed based on the $\vert S +1\vert $-th smallest eigenvalue of the Laplacian matrix $\LL$, the minimal degree of nodes in set $V \setminus S$, and the number of edges connected nodes in $S$ and $V \setminus S$. Properties for the smallest eigenvalue of grounded Laplacian matrix of weighted undirected~\cite{MaBe17} and directed~\cite{XiCa17} also received attention from the  scientific community.

The smallest eigenvalue $\lambda(S)$ of matrix $\LL(S)$ captures the importance of nodes in set $S$ as a whole in graph $\calG$, via the convergence rate of leader-follower systems~\cite{RaJiMeEg09}, the effectiveness of pinning scheme for pinning control of networks~\cite{LiXuLuChZe21}, and so on~\cite{PiShFiSu18}. We thus use $\lambda(S)$ to quantify the importance/centrality of the group of nodes in $S$, termed grounded node group centrality. There are numerous metrics for centrality of a group of nodes in a graph~\cite{VaAnMe20}, based on structural or dynamical properties, including betweenness~\cite{DoelPuZi09,MaTsUp16}, closeness centrality~\cite{EvBo99,BeGoMe18}, and current flow closeness centrality~\cite{LiPeShYiZh19}, and so on. However, since the criterion for importance of a set of nodes depends on specific applications~\cite{GhTeLeYa14}, grounded node group centrality deviates from previous centrality metrics, even for an individual node~\cite{PiSu14}. Therefore, it is interesting and necessary to study the grounded node group centrality.

Various previous work~\cite{WaLiXuLi19, LiXuLuChZe21} focused on selecting a set $S$ of $k=\vert S\vert $ grounded nodes in order to maximize the smallest eigenvalue $\lambda(S)$ of the grounded Laplacian matrix $\LL(S)$ for specific purposes, such as maximizing the convergence rate of leader-follower systems~\cite{ClAlBuPo12,ClBuPo12} and the effectiveness of pinning control~\cite{LiXuLuChZe21}. A submodular approach with polynomial-time complexity was developed in~\cite{ClHoBuPo18}, while a feature-embedded evolutionary algorithm was presented in~\cite{ZhTa20}. These papers proposed different greedy heuristic algorithms for this combinatorial optimization problem, but left its computational complexity open. Moreover, the properties of the objective function for the problem are still  not well understood. For example, it is not known whether the objective function is submodular or not. Finally, at present, there is no nearly linear time algorithm with good performance for this problem. Thus, almost all experiments in existing studies were executed in medium-sized networks with up to 1,000 nodes.  



\section{Preliminaries}
\label{section:pre}


This section briefly introduces some notations, basic concepts about a graph, its associated matrices, especially the grounded Laplacian matrix, and submodular function.

\subsection{Notations}

Throughout this paper, scalars in $\rea$ are denoted by normal lowercase letters like $a,b,c$, vectors by bold lowercase letters like $\aaa,\bb,\cc$, sets by normal uppercase letters like $A,B,C$, and matrixes by bold uppercase letters like $\AA,\BB,\CC$. 
Notations $\aa^\top$ and $\AA^\top$ are used to denote the transpose of a vector $\aa$ and a matrix $\AA$, respectively. Symbol $\aa_i$ is used to denote the $i$-th element of  vector $\aa$, and $\AA_{ij}$ is used to denote the entry $(i,j)$ of matrix $\AA$. In addition,  use $\boldsymbol{1}$ to denote the  vector of appropriate dimension, whose elements are all ones,  use $\ee_i$ to denote a vector with $i$-th element being 1 and others 0, and use $\EE_{ij}$ to denote an appropriate dimension matrix with the entry $(i,j)$ being 1 and others being 0.



\subsection{Graph and  Grounded Laplacian Matrix}

We denote an undirected connected binary graph as $\calG=(V,E)$, where $V$ is the set of nodes/vertices with size $\vert V\vert =n$ and $E\subseteq V \times V$ is the the set of edges with size $\vert E\vert =m$. Let $\AA \in \{0,1\}^{n \times n}$ be the adjacency matrix for graph $\calG$ where $\AA_{ij}$ equals 1 if nodes $i$ and $j$ are directly connected by an edge and 0 otherwise. The neighbors of node $i\in V$ in graph $\calG$ are given by the set $\calN_{{i}} =\{j \in V \vert  (i, j) \in E\}$. The degree of node $i$ is represented as $d_{i}=\vert \calN_{i}\vert $ and the degree matrix is defined as $\DD=\text{diag}(d_{1},d_{2},\dots ,d_{n})$ accordingly. The Laplacian matrix for the graph is given by $\LL=\DD-\AA$, which is a positive semi-definite matrix .



For a graph $\calG=(V,E)$, its grounded Laplacian matrix $\LL(S)$ is a variant of Laplacian matrix $\LL$, which is obtained from $\LL$ by deleting those rows and columns corresponding to the nodes in a nonempty set $S$ of size $k$~\cite{Mi93}. The nodes in $S$ are called grounded nodes. By definition, $\LL(S)$ is a $(n-k) \times (n-k)$ principal submatrix of Laplacian matrix $\LL$, which is a symmetric diagonally-dominant M-matrix (SDDM). From~\cite{McNeScTs95}, $\LL(S)$ is a positive definite matrix, and its inverse is a non-negative matrix. Let $\lambda(S)$ denote the smallest eigenvalue of $\LL(S)$, which is a monotonic increasing function of set $S$~\cite{PaBa10}. By Perron–Frobenius theorem~\cite{Ma00}, $\lambda(S)$ has an eigenvector $\uu$ with non-negative components.





\subsection{Submodular Function}

We continue to introduce  monotone non-decreasing and submodular set functions. For a set $S$, we use $S+u$ to denote $S\cup\{u\}$.
\begin{definition}[Monotonicity]\label{def:mono}
A set function $f:2^{V}\rightarrow \mathbb{R}$ is monotone non-decreasing if $f(S) \le f(T)$ holds for all $S \subseteq T \subseteq V$.
\end{definition}

\begin{definition}[Submodularity]\label{def:sub}
A set function $f:2^{V}\rightarrow \mathbb{R}$ is submodular if
\begin{equation*}\label{submod}
f(S + u) - f(S) \geq f(T + u) - f(T)
\end{equation*}
 holds for all
$S \subseteq T \subseteq V$ and $u \in V$.
\end{definition}



\section{PROBLEM FORMULATION}
\label{section:formulation}

The smallest eigenvalue $\lambda(S)$ of grounded Laplacian matrix $\LL(S)$ directly connects to a wide range of applications. As shown above, $\lambda(S)$ characterizes the performance of various systems, including the convergence rate of a leader-follower system~\cite{RaJiMeEg09}, the effectiveness of  pinning control~\cite{LiXuLuChZe21}, robustness to disturbances of the leader-follower system via the $\mathcal{H}_\infty$~\cite{PiShFiSu18}, and so on. For these systems, larger $\lambda(S)$ corresponds to better performance. 

\subsection{Problem Statement}


As shown in~\cite{PaBa10}, $\lambda(S)$ is a monotonic increasing function of set $S$. This motivates us to study the problem of how to select $k$ nodes to form the set of grounded nodes $S$, so that the eigenvalue $\lambda(S)$ is maximum.

\begin{problem}[\underline{Max}imization of the \underline{S}mallest \underline{E}igenvalue of \underline{G}rounded \underline{L}aplacian (MaxSEGL)]\label{prob:2}
Given an unweighted, undirected, and connected network $\calG=(V,E)$, we aim to find a subset $S\subset V$ with $0<k  \ll \vert V\vert $ nodes so that the smallest eigenvalue of the grounded Laplacian matrix $\LL(S)$ is maximized. The problem can be formulated as
\begin{equation*}
  	 S^*=\arg\max_{S\subset V,\vert S\vert = k} \lambda(S).
\end{equation*}
\end{problem}

Problem~\ref{prob:2} is a combinatorial optimization problem. We can naturally think of a direct solution method by exhausting all $\tbinom{n}{k}$ cases for set $S$. For each possible case of set $S$ containing $k$ nodes, calculate the smallest eigenvalue $\lambda(S)$ of the grounded Laplacian matrix $\LL(S)$ induced by the grounded nodes in set $S$. Then, output the optimal solution $S^*$, which maximizes the smallest eigenvalue for grounded Laplacian matrix with $k$ grounded nodes. Due to the exponential complexity $O(\tbinom{n}{k}m)$, this algorithm fails when $n$ or $k$ is slightly oversized.


\subsection{Grounded Node Group Centrality}

As is well known, the smallest eigenvalue $\lambda(S)$ characterizes the performance of various dynamical systems, which is determined by grounded nodes in $S$. In this sense, $\lambda(S)$ can be used as a centrality of the group of nodes in $S$, called grounded node group centrality. The larger the value of $\lambda(S)$, the more important the nodes in group $S$ for related dynamic systems. Thus, Problem~\ref{prob:2} is precisely the problem of finding the $k$ most important nodes forming grounded node set $S$, such that $\lambda(S)$ is maximum. Note that when set $S$ includes only one node $i$, $\lambda(\{i\})$ is reduced to the grounded centrality of an individual node defined in~\cite{PiSu14}.

It is worth mentioning that for a group $S$ of nodes, its grounded node group centrality $\lambda(S)$ does not often equal to the integration of the centrality of individual nodes in set $S$, due to the interdependence of the nodes in set $S$. In other words, $\lambda(S) \neq \sum_{i \in S}\lambda(\{i\})$. Thus, the combinatorial enumeration in Problem~\ref{prob:2} is generally not appropriately solved by choosing the top-$k$ nodes with the largest individual grounded centrality values. For example, for the line graph in Fig.~\ref{fig:nosub}, the grounded centrality for nodes 1-7 are, 0.058, 0.081, 0.121, 0.198, 0.121, 0.081 and 0.058. If $S$ consists of three nodes, is easy to obtain that one of the optimal set $S^*$ is $\{2,4,6\}$, instead of $\{3,4,5\}$ whose elements are top-$3$ nodes with the largest grounded centrality scores.

\begin{figure}[htbp]
    \centering
    \includegraphics[width=0.6\textwidth]{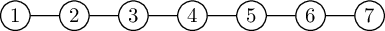}
    \caption{A line graph  with 7 nodes and 6 edges.}
    \label{fig:nosub}
\end{figure}

It has been demonstrated in~\cite{PiSu14} that for individual nodes, the grounded centrality gives quite a different  ranking for node  importance, compared with many other node centrality measures, such as degree centrality, betweenness centrality, eigenvector centrality, closeness centrality, and information centrality. Thus, grounded node group centrality also deviates from previously proposed centrality measures for a node group, including betweenness centrality~\cite{DoelPuZi09,MaTsUp16}, closeness centrality~\cite{EvBo99,BeGoMe18}, as well as current flow closeness centrality~\cite{LiPeShYiZh19}. Therefore, it is of independent interest to study Problem~\ref{prob:2}, with an aim to find the $k$ most important nodes as the set $S^*$ of grounded nodes corresponding to maximum $\lambda(S^*)$.

\subsection{Non-submodularity of the Objective Function}




Given the combinatorial nature of Problem~\ref{prob:2}, one seeks a heuristic approximately solving this problematic combinatorial problem. For a combinatorial optimization problem, finding appropriate properties of its objective function is crucial to solving it effectively. For example, when the objective function is submodular, a simple greedy algorithm by selecting one element with maximal marginal benefit in each iteration yields a solution with $1-e^{-1}$ approximation ratio~\cite{NeWoFi78}, which has been widely used in combinatorial optimization problems and has effectively solved many NP-hard problems in recent years.

However, there are still a lot of combinatorial optimization problems whose objective functions are not submodular. For these problems, a heuristic cannot guarantee a $(1-e^{-1})$-approximation solution. Unfortunately, Problem~\ref{prob:2} happen to be this problem class. To show the non-submodularity of the objective function for Problem~\ref{prob:2}, we give an example of the line graph with 7 nodes and 6 edges in Fig.~\ref{fig:nosub}.

Let set $A$ = $\{1\}$, $B$ = $\{1,2\}$ and $v$ = $6$. Simple computation leads to
\begin{align*}
\lambda(A)=0.0581, & &  \lambda(A+v)=0.3820, \\
\lambda(B)=0.0810, & &  \lambda(B+v)=0.5858,
\end{align*}
which means
\begin{equation*}
\lambda(A+v)-\lambda(A)=0.3239<0.5048=\lambda(B+v)-\lambda(B).
\end{equation*}
Therefore, the objective function of Problem~\ref{prob:2} is non-submodular.


\section{Hardness of  Problem~\ref{prob:2} and a  Na\"{\i}ve Greedy Heuristic Algorithm}

In this section, we prove that Problem~\ref{prob:2} is NP-hard, and then provide a  na\"{\i}ve greedy heuristic algorithm for this problem.

\subsection{Hardness of  Problem~\ref{prob:2}}

As shown above,  Problem~\ref{prob:2} is combinatorial, and thus can be exactly solved by brute-force search. On the other hand, the objective function is the smallest eigenvalue of a grounded Laplacian matrix,  which suggests that Problem~\ref{prob:2} is likely to be very difficult.  Here we confirm this intuition  by proving that Problem~\ref{prob:2} is  NP-hard, by  reduction from node cover on 3-regular graphs, which has been proved to be  an NP-complete problem~\cite{FrHeJa98}. A $3$-regular graph is a graph for which the degree of every node is $3$.
For a graph $\calG=(V,E)$,  a node set $C \subset V$  is called a node cover of  graph $\calG$ if every edge in $E$ has at least one endpoint in  $C$.   A node set $I$ is called an independent set of   graph $\calG$, if no two nodes in  which are adjacent. By definition,  if $C$ is a node cover of $\calG$, then $V \setminus C$ is an independent set of $\calG$.

The decision version of vertex cover on a connected 3-regular graph  is stated below.

\begin{problem}[\uline{V}ertex \uline{C}over on a \uline{3}-\uline{R}egular \uline{G}raph, VC3RG]
Given a connected 3-regular graph $\calG=(V,E)$ and a positive integer $k$, decide whether or not there exists a node  set $S \subset V$ such that $\vert S\vert = k$ and $S$ is a node cover of $\calG$.
\end{problem}

We proceed to give the decision version of Problem~\ref{prob:2}.

\begin{problem}[MaxSEGL \uline{D}ecision Version, MaxSEGLD]
Given a connected graph $\calG=(V,E)$, a positive integer $k$, and a postive real number $r \in \mathbb{R}^+$, decide whether or not there exists a node set $S \subset V$ such that $\vert S\vert = k$ and $\lambda(S)\geq r$.
\end{problem}

Before giving the reduction, we introduce the following lemma.

\begin{lemma}\label{lem:reduc}
Let $\calG=(V,E)$ be a connected 3-regular graph, and let $S$ be a nonempty subset of $V$. 
Then for the grounded Laplacian $\LL(S)$ with nodes in $S$ being grounded nodes, $\lambda(S)\leq 3$ with equality if and only if $S$ is a node cover of $\calG$.
\end{lemma}
\begin{proof}
If $S$ is a node cover of $\calG$,  then $V \backslash S$ is an independent set of $\calG$, which implies that there is no edge between any pair of nodes in $V\backslash S$. Thus,  $\LL(S)$ is a diagonal matrix with all diagonal entries being 3, leading to $\lambda(S)=3$.

We continue to show that if $S$ is not a node cover of $\calG$ then $\lambda(S)<3$. When set $S$ is not a node cover of $\calG$, then $V\backslash S$ is definitely  not an independent set of $\calG$. Thus, $\LL(S)$ can be written as  a block diagonal matrix, with each block corresponding to a connected component of a subgraph $\calG[V\backslash S]$, which is  induced by nodes in $V \backslash S$. Let $\mathbf{T}$ be a block of $\LL(S)$, which corresponds to $t> 1$ nodes in set $V\backslash S$.  The diagonal elements of matrix $\mathbf{T}$  is 3. Suppose that for the non-diagonal elements in $\mathbf{T}$, $q>0$ entries have values of $-1$, while the others are 0.

Let $\lambda_{\min}(\TT)$ be the smallest eigenvalue of the matrix $\TT$. Then for the vector $\boldsymbol{1}$, we have
\begin{equation}
\lambda(S) \leq \lambda_{\min}(\TT) \leq \frac{\boldsymbol{1}^\top \TT \boldsymbol{1}}{\boldsymbol{1}^\top \boldsymbol{1}} = \frac{3t-q}{t} <3,
\end{equation}
which completes the proof.
\end{proof}

By Lemma~\ref{lem:reduc} we obtain the following theorem.

\begin{theorem}
Maximizing the smallest eigenvalue of a grounded Laplacian matrix subject to a cardinality constraint is NP-hard.
\end{theorem}
\begin{proof}
We first provide a polynomial reduction
\begin{equation}
p: \{\calG=(V,E),k\} \to  \{\calG=(V,E),k,r\}
\end{equation}
from instances of VC3RD to instances of  MaxSEGLD. For a connected 3-regular graph $\calG = (V, E)$ with $n$ nodes, we construct the following reduction $p$ as
\begin{equation}
p((\calG=(V,E),k))=(\calG=(V,E),k,3).
\end{equation}
By Lemma~\ref{lem:reduc}, $p$ is a polynomial reduction from VC3RD to MaxSEGLD, indicating  that Problem~\ref{prob:2} is NP-hard.
\end{proof}

\subsection{A Na\"{\i}ve Heuristic Algorithm}


Due to the NP-hardness of Problem~\ref{prob:2}, we resort to efficient heuristics for this problem. First, we propose a na\"{\i}ve greedy heuristic algorithm. As shown in~\cite{PaBa10}, given a graph $\calG=(V,E)$, a set $S\subset V$ of grounded nodes, and a node $j \in V \backslash S$, $\lambda(S+j) \geq \lambda(S)$ holds by Cauchy's interlacing theorem~\cite{Ba10}. Define $\lambda_S(j)=\lambda(S+j)-\lambda(S)$ as the increase of the smallest eigenvalue after node $j$ is added to set $S$ of grounded nodes, then $\lambda_S(j) \geq 0$ for any node $j \in V \backslash S$. Then, the na\"{\i}ve greedy heuristic algorithm, outlined in Algorithm~\ref{alg:exact}, is presented as follows. Initially, the set $S$ of grounded nodes is set to be empty. Then $k$ nodes are iteratively picked to $S$ as grounded nodes from set $V \backslash S$ of candidate nodes. In each iteration of the na\"{\i}ve greedy algorithm, the node $j$ in candidate node set is selected to maximize the quantity $\lambda_S(j)$. The na\"{\i}ve algorithm stops when $k$ nodes are selected to be added to set $S$.

\begin{algorithm}[htbp]
  \caption{\textsc{Na\"{\i}ve}$(\calG, k)$}
	\label{alg:exact}
\begin{algorithmic}
\Require A graph $\calG=(V,E)$; an integer $k < \vert V\vert $
\Ensure $S$: a subset of $V$ with $\vert S\vert  = k$
\State Initialize solution $S = \emptyset$
\For{$i = 1$ to $k$}
		\State Compute $\lambda_S(j) = \lambda(S +j)-\lambda(S)$ for each $j \notin S$
		\State Select $s$ s.t.  $s \gets \mathrm{arg\, max}_{j \in V \setminus S} \lambda_S(j)$
		\State Update solution $S \gets S \cup \{s\}$
\EndFor
\State \Return $S$
\end{algorithmic}
\end{algorithm}

Although there is no guarantee for the error between this na\"{\i}ve method and the optimal solution due to the posteriority and non-submodularity, in the Experiments Section, we will see that this method has many advantages over some baseline schemes. Note that the main computation cost for the na\"{\i}ve algorithm lies in determining $\lambda_S(j)$ in each of $k$ iterations. For small $k$, for each $j \in V \setminus S $, the increment of the smallest eigenvalue $\lambda_S(j) $ contributed by node $j$ can be calculated by the numerical method in $O(m)$ time~\cite{La52}. Thus, a direct implementation of this na\"{\i}ve approach takes $O(knm)$ time.

\section{Nearly Linear Time Approximation Algorithm}
\label{section:algorithm}

Although the na\"{\i}ve approach described in Algorithm~\ref{alg:exact} is much faster than  brute-force search, this algorithm is still infeasible for large-scale networks with millions of nodes due to its high computational requirements. The main computational cost of Algorithm~\ref{alg:exact}  lies in determining the increment $\lambda_S(j)$ of the smallest eigenvalue of grounded Laplacian matrix when a new node $j$ is added into the set $S$ of grounded nodes. In this section, we propose a nearly linear time approximation algorithm evaluating $\lambda_S(j)$ for every candidate node $j \in V \setminus S$.

Note that as a principal submatrix of Laplacian matrix $\LL$, the grounded Laplacian matrix $\LL(S)$ is a SDDM matrix, and its inverse $\LL(S)^{-1}$ is a non-negative matrix. By matrix tree theorem~\cite{Ch82}, when an edge linking two nodes in $V \setminus S$ is removed, the determinant of matrix $\LL(S)$ decreases. Thus, the determinant  of matrix $\LL(S)^{-1}$ increases, it is the same with  $\lambda (S)$ that is the largest eigenvalue of matrix $\LL(S)^{-1}$. Thus, $\lambda (S)$ is an increasing function of each edge in $ E \cap ((V \backslash S) \times (V \backslash S))$.  For an edge linking two nodes in $S$ and $V \backslash S$, its removal often does not lead to the increase of $\lambda (S)$~\cite{PiShFiSu18}. Moreover, for a node $j$ with large $\lambda_S(j)$ in $V \backslash S$ in a large network, the number of its neighbors in $S$ is relatively small. Recall that picking a new node $j$ added to $S$ is equivalent to removing the row and column from $\LL(S)$ corresponding to node $j$. Thus, $\lambda_S(j)=\lambda(S+j)-\lambda(S)$ can be approximately regarded as the sum of displacement of $\lambda(S)$ caused by the removal of  edges incident to node $j$, with the other ends in $V \backslash S$. In the sequel, we use two different approaches to evaluate the sensitivity $\lambda_S(j)$ of eigenvalue $\lambda (S)$, when node $j$ is added to $S$.



\subsection{Derivative Based  Eigenvalue Sensitivity}

The increment $\lambda_S(j)$ of the smallest eigenvalue $\lambda(S)$ can be looked upon as the sum of increase of $\lambda(S)$ contributed by the removal of every edge incident to node $j$. For convenience of the following presentation, we use $\lambda$ to represent $\lambda(S)$. Moreover, we introduce a centrality measure $\lambda_S(e)$, which quantifies the impact of removal of an edge $e$ on the increase of the smallest eigenvalue $\lambda(S)$. In~\cite{YiSh2018,SiBoBaMo18,KaTo19}, the partial derivative was applied to measure the role of an edge on certain quantities concerned. Motivated by these previous work, we employ a similar idea to characterize $\lambda_S(e)$, defined as the rate at which $\lambda(S)$ changes with respect to the change of the weight of edge $e$, since an weighted graph can be regarded as a weighted one, with each edge being a unit weight.

Specifically, for an edge $e$ with two end nodes $i$ and $j$, $\lambda_S(e)$ is defined by derivative matrix as $\lambda_S(e)=\lambda_S(i,j)=\frac{\partial \lambda}{\partial \LL(S)_{ij}}$, since $\lambda$ is a function of the weight on every edge. Lemma~\ref{lem:NBDC} shows that $\lambda_S(e)$ can be expressed in terms of the product of elements corresponding to two end nodes $i$ and $j$ of the eigenvector $\uu$ associated with $\lambda$.



\begin{lemma}\label{lem:NBDC}
Given a graph $\calG=(V,E)$ with Laplacian matrix $\LL$, let $S \subset V$ be a  set of grounded nodes, and let $(\lambda,\uu)$ be the eigen-pair corresponding to  the smallest eigenvalue and its associated eigenvector for the grounded Laplacian matrix $\LL(S)$. Then for any edge $e=(i,j) \in E \cap ((V \backslash S) \times (V \backslash S))$, we have
\begin{equation*}
\lambda_S(e)=\lambda_S(i,j)=\frac{\partial \lambda}{\partial \LL(S)_{ij}}=\uu_i\uu_j.
\end{equation*}

\begin{proof}
By definition of the eigenvalues $\lambda$ and eigenvector $\uu$, we have $\LL(S)\uu=\lambda\uu$. Performing the partial derivative with respect to $\LL(S)_{ij}$ on both side leads to
\begin{equation*}
\frac{\partial \LL(S)}{\partial \LL(S)_{ij}}\uu +\LL(S)\frac{\partial \uu}{\partial \LL(S)_{ij}}=\frac{\partial \lambda}{\partial \LL(S)_{ij}}\uu + \lambda\frac{\partial \uu}{\partial \LL(S)_{ij}}.
\end{equation*}
Left multiplying $\uu^\top$ on both sides of the above equation yields $\lambda_S(e)=\frac{\partial \lambda}{\partial \LL(S)_{ij}}=\uu_i\uu_j$.
\end{proof}
\end{lemma}

By definition $\lambda_S(e) $ measures the importance of  edge $(i,j)$ or  $(j,i)$ with $i$ and $j$ being the two end nodes of $e$, which is equal to the sensitivity of $\lambda(S)$ to the change in the weight of edge $e$.  Intuitively, for any node $j$ in $V\setminus S$, $\lambda_S(j)$ is naturally equal to the sum of the importance of the edges adjacent to $j$. Hence, the change rate of $\lambda(S)$ after removing from matrix $\LL(S)$ the row and column corresponding to node $j$ can be evaluated as
\begin{equation}\label{eq:c1}
\tilde{\lambda}_S(j)=2\,\sum_{i \in \calN_j\setminus S}\lambda_S(i,j)=2 \uu_j\sum_{i \in \calN_j \setminus S} \uu_i,
\end{equation}
since for every edge  $e$  with  end nodes $i$ and $j$,  $(i,j)$ and  $(j,i)$ are both counted.

\subsection{Matrix Perturbation Based  Eigenvalue Sensitivity}



In this subsection, we show that the increment $\lambda_S(j)$ of the smallest eigenvalue $\lambda(S)$ can also be evaluated by matrix perturbation theory~\cite{Mi11}, and the matrix perturbation technique has been widely used for eigenvalues~\cite{HeYaYuZh19,MiSuNi10,ReOtHu06,BoPaRa99}. To this end,  we first estimate the increment of the smallest eigenvalue by using matrix perturbation, when one edge is removed.

Notice that for a graph $\calG=(V,E)$ with grounded node set $S\subset V$, if an edge  $e$ with end nodes $i$ and $j$ in $V\setminus S$ is deleted, the resultant  grounded Laplacian becomes $\LL(S) + \EE_{ij}+\EE_{ji}$. Let $(\lambda,\uu)$ be the eigen-pair corresponding to  the smallest eigenvalue and its associated eigenvector for the grounded Laplacian matrix $\LL(S)$. Define  $\Delta \LL(S)=\EE_{ij}+\EE_{ji}$.
Then if one  perturbs the matrix $\LL(S)$ with $\Delta \LL(S)$, $\lambda$ varies with $\Delta \lambda$ and $\uu$ varies with $\Delta \uu$ after perturbing. Thus, the eigenequation after removing edge  $e$ is given by
\begin{equation}\label{equa:pert}
    (\LL(S)+\Delta \LL(S))(\uu+\Delta\uu)=(\lambda +\Delta \lambda)(\uu+\Delta\uu),
\end{equation}
Left multiplying $\uu^\top$ on both sides of Eq.~\eqref{equa:pert} yields
\begin{equation*}
\uu^\top \Delta \LL(S) \uu +\uu^\top \Delta \LL(S) \Delta\uu=\uu^\top \Delta \lambda \uu +\uu^\top \Delta \lambda \Delta\uu.
\end{equation*}
Then the eigen-gap for the smallest eigenvalue is
\begin{equation*}
\Delta \lambda = \frac{\uu^\top \Delta \LL(S) \hat{\uu}}{\uu^\top \hat{\uu}}=\frac{\uu_i\hat{\uu}_j+\uu_j\hat{\uu}_i}{\uu^\top \hat{\uu}},
\end{equation*}
where $\hat{\uu}=\uu+\Delta \uu$.

For a large graph network $\calG$, the removal of one edge has little influence on the whole network, and also does not influence the eigenvector $\uu$ too much,  meaning $\Delta \uu \approx \boldsymbol{0}$~\cite{HeYaYuZh19,MiSuNi10,ReOtHu06}. Hence
\begin{equation*}
 \Delta \lambda \approx 2\uu_i\uu_j>0,
\end{equation*}
which is consistent with Lemma~\ref{lem:NBDC}.


The following lemma is immediate by repeating the operation of edge removal on all the edges incident to a node $j$.

\begin{lemma}\label{lem:deltall}
Let $\calG=(V,E)$ be a connected graph and $\LL$ be its Laplacian matrix, let $S \subset V$ be a node set. For the matrix $\LL(S)$, we define $(\lambda,\uu)$ be the eigenpair corresponding to  the smallest eigenvalue $\lambda$ and its associated eigenvector $\uu$ for the grounded Laplacian matrix $\LL(S)$. For any node $i \in V \backslash S$, define $\Delta \bar{\LL}(S)=\sum_{j\in \calN_i\setminus S}(\EE_{ij}+\EE_{ji})$. For the case the matrix $\LL(S)$ is perturbed with $\Delta \bar{\LL}(S)$, suppose $\lambda$ varies with $\Delta \bar{\lambda}$ and $\uu$ varies with $\Delta \bar{\uu}$. Then,
\begin{equation*}
    (\LL(S)+\Delta \bar{\LL}(S))(\uu+\Delta\bar{\uu})=(\lambda +\Delta \bar{\lambda})(\uu+\Delta\bar{\uu}),
\end{equation*}
and the eigen-gap for the smallest eigenvalue can be approximated by
\begin{equation*}
 \Delta \bar{\lambda} \approx 2\,\uu_j\sum_{i \in \calN_j \setminus S} \uu_i=\tilde{\lambda}_S(j),
\end{equation*}
as given by Eq.~\eqref{eq:c1}.
\end{lemma}


According to Lemma~\ref{lem:deltall}, if we add a new node $j$ into set $S$, then $\Delta \bar{\lambda}$ equals to $\lambda_S(j)$ when $\lambda(S+j) \leq d_j$. Thus maximizing $\lambda(S+j)$ is reduced to maximizing $\lambda_S(j)$ or its approximation $\tilde{\lambda}_S(j)$.



\begin{remark}
If $S$ is an empty set, it is apparent that the smallest eigenvalue of $\LL(S)$ is 0, and the corresponding eigenvector is $\uu=\boldsymbol{1}$. Hence the importance of each node equals to its degree, which is consistent with the conclusion in~\cite{LiXuLuChZe21} that selecting nodes with large degrees is a good solution when only choosing a few number of grounded nodes.
\end{remark}

\subsection{Fast Algorithm}


The above analysis  shows that the optimization of $\lambda(S)$ is reduced to the computation of the corresponding eigenvector $\uu$ at each iteration. Below we introduce two methods, Lanczos method and SDDM solver or Laplacian solver, for computing the eigenvector $\uu$, and compare them in Experiments Section.


Lanczos method~\cite{La52} is classical method to calculate the eigenvalues and eigenvectors of a Hermitian matrix by iteratively tridiagonalizing the matrix. Although it was  proposed almost 70 years ago,  it still remains  one of the most important algorithms in numerical computation, and has become a standard algorithm for evaluating eigenvalues and eigenvectors. Since in our case only the eigenpair associated with the smallest eigenvalue of the grounded Laplacian matrix is needed,  Lanczos method outputs the results in a very short  time.

SDDM solver or Laplacian solver~\cite{BaSpSrTe13,SpTe14,CoKyMiPaJaPeRaXu14} is another method to  estimate eigenvector $\uu$ without  computing the complete eigensystem. By combining  SDDM solver and the inverse power method, the eigenvector of the  smallest eigenvalue for a SDDM matrix $\MM$ can be solved in nearly linear time. For consistency,  we first introduce the SDDM solver.

\begin{lemma}\label{lem:solver}
There is a nearly linear time solver $\xx = \SDDMSolver(\MM, \yy, \epsilon)$ which takes a symmetric positive semi-definite matrix $\MM_{n\times n}$ with $m$ nonzero entries, a vector $\bb \in \mathbb{R}^n$, and an error parameter $\delta > 0$, and returns a vector $\xx \in \mathbb{R}^n$ satisfying $\norm{\xx - \MM^{-1} \yy}_{\MM} \leq \delta \norm{\MM^{-1}\yy}_{\MM}$ with high probability, where $\norm{\xx}_{\MM} \defeq \sqrt{\xx^\top \MM \xx}$. The solver runs in expected time $\tilde{O}(m)$, where $\tilde{O}(\cdot)$ notation suppresses the ${\rm poly} (\log n)$ factors. 
\end{lemma}

\begin{algorithm}[htbp]
	\caption{\textsc{ApproxVector}$(\MM,\epsilon)$}\label{alg:vector}
\begin{algorithmic}
		\Require A SDDM matrix $\MM$; an error parameter $0<\epsilon<1$
		\Ensure $\uu$: an approximation of the eigenvector of the smallest  eigenvalue  of  matrix $\MM$
		\State Initialize vector $\vv=\boldsymbol{1}^\top$
        \State Set $\omega=1$
		\While{$\omega>\epsilon$}
            \State $\uu=\vv$
            \State $\vv = \SDDMSolver(\MM, \uu, \epsilon)$
            \State $\vv=\frac{\vv}{\norm{\vv}}$
            \State $\omega=\vert \uu^\top \MM \uu - \vv^\top \MM \vv\vert /\vv^\top \MM \vv$
		\EndWhile
\State \Return $\uu$
\end{algorithmic}
\end{algorithm}

Algorithm~\ref{alg:vector} describes the details for computing  the eigenvector $\uu$ associated with the smallest eigenvalue of a grounded Laplacian matrix, which combines the SDDM solver and the inverse power method (line 5). It avoids directly inverting the matrix $\MM$ and significantly improves the efficiency of calculation. Since the setting of parameters is not the focus of our study in the approximation calculation, we simply set the error parameter in the SDDM Solver in line 5 to be $0<\epsilon<1$, while specific details can be found in~\cite{CoKyMiPaJaPeRaXu14}. Also, in most cases the smallest eigenvalue of a grounded Laplacian matrix is very small,  very few iterations, even one iteration, can return a fairly good approximation for $\uu$.


After obtaining  the eigenvector $\uu$ of the smallest eigenvalue  $\lambda(S)$ of the grounded Laplacian matrix, by using  Eq.~\eqref{eq:c1}, for each node $j \in V \backslash S$, $\tilde{\lambda}_S(j)$ can be directly computed. Having these as a prerequisite, we now propose a $\tilde{O}(km)$-time approximation algorithm \textsc{Fast} for the MaxSEGL problem formulated in Problem~\ref{prob:2}. In Algorithm~\ref{alg:approx}, we  present the computation details, which runs $k$ iterations.  In each iteration, the calculation of the  eigenvector $\uu$ (line 4) takes $\tilde{O}(m)$ running time, and the computation of $\tilde{\lambda}_S(j)$ (line 5) for each node $j \in V \backslash S$ takes $O(m)$ time. Thus, the total running time for Algorithm \textsc{Fast} is $\tilde{O}(km)$, which is applicable to large  networks with millions of nodes, as will be shown in the next section.

\begin{algorithm}[htbp]
	\caption{\textsc{Fast}$(\calG, k,\epsilon)$}\label{alg:approx}
\begin{algorithmic}
		\Require A graph $\calG=(V,E)$; an integer $k < \vert V\vert $; an error parameter $\epsilon>0$
		\Ensure $S$: a subset of $V$ with $\vert S\vert  = k$
		\State Initialize solution $S = \emptyset$
        \State Let $\LL$ be the Laplacian matrix of graph $\calG$
		\For{$i = 1$ to $k$}
           \State $\uu=\textsc{ApproxVector}(\LL(S),\epsilon)$
           \State Compute $ \tilde{\lambda}_S(j)= 2\,\uu_j\sum_{i \in \calN_j \setminus S} \uu_i$ for each $j \notin S$
			\State Select $s$ s.t.  $s \gets \mathrm{arg\, max}_{j \in V \setminus S} \tilde{\lambda}_S(j)$
			\State Update solution $S \gets S \cup s$
		\EndFor
\State \Return $S$
\end{algorithmic}
\end{algorithm}



\section{Experiments}
\label{experiments}

In this section,  we perform extensive experiments to evaluate the  performance of our proposed
algorithms on diverse real-world networks, in terms of effectiveness and efficiency.

\subsection{Setup}

\textbf{Datasets}. Our algorithms are tested on a diverse set of real-world networks with  up to millions of nodes, all of which are publicly available in KONECT~\cite{Ku13} and SNAP~\cite{LeSo16}. Without loss of generosity, we only consider  connected networks, while for any disconnected network we  execute experiments on its largest component.  Relevant statistics of these datasets is summarized in Table~\ref{table:networksize}, where networks are shown in increasing order of the numbers of nodes.

\begin{table}[htbp]\centering
\caption{Information about the networks with $n$ nodes and $m$ edges.}\label{table:networksize}
\begin{tabular}{ScSS} \toprule
  & {\textbf{Network}}&$n$ & $m$\\ \midrule
1  & Tribes & 16 & 58 \\
2 & Firm-Hi-Tech & 33 &147\\
3 & Karate & 34 &78 \\
4 & Dolphins     &62 &159 \\

5 & 685-bus       & 685        & 1282       \\
6 & Email-Univ       & 1133       & 5451       \\
7 & Bcspwr09      & 1723       & 2394       \\
8 & Routers-RF     & 2113       & 6632  \\
9 & US-Grid       & 4941       & 6594       \\
10 & Bcspwr10      & 5300       & 8271       \\
11 & Pages-Government & 7057       &  89455 \\
12 & WHOIS          & 7476       & 56943      \\

13 & Pretty Good Privacy     & 10680      & 24340      \\
14 & Anybeat         & 12645      & 67053      \\
15 & Webbase-2001    & 16062      & 25593      \\
16 & As-CAIDA2007   & 26475      & 53381 \\

17 & Epinions        & 26588      & 100120   \\
18 &  Email-EU         & 32430      & 54397      \\
19 & Internet-As    & 40164      & 85123     \\
20 & P2P-Gnutella   & 62561      & 147878  \\
21 & RL-Caida    & 190914     & 607610     \\
22  &  DBLP-2010  & 226413& 716460    \\
23  & Twitter-follows &404719 &713319  \\
24  &Delicious &536108 &1365961  \\
25    &FourSquare & 639014 & 3214986  \\
26  &Youtube-Snap &1134890 & 2987624  \\
    \bottomrule
\end{tabular}
\end{table}

\textbf{Machine Configuration and Reproducibility}. All our algorithms are programmed and  implemented in Julia, in order to facilitate the use of the SDDM solver  in the Julia Laplacian.jl package  available on the website\footnote{https://github. com/danspielman/Laplacians. jl}. The error parameter $\epsilon$ is set to be $10^{-3}$ in the experiments. The Lanzcos method is also programmed in Julia. All our experiments are run on a Linux box with 4.2 GHz Intel i7-7700 CPU and 32G memory, using a single thread.

\textbf{Baseline Methods.} Our proposed algorithms,  Algorithm \textit{Na\"{\i}ve} and Algorithm \textit{Fast}
(that is, Algorithm~\ref{alg:exact} and Algorithm~\ref{alg:approx}), are compared with several baseline methods,  which are briefly summarized as follows.
\begin{enumerate}
    \item \textsc{Optimum}: choose $k$ nodes forming  the set $S$ that maximize $\lambda(S)$ by exhaustive search.
	\item \textsc{Degree}: choose $k$ nodes with  largest degrees.
	\item \textsc{Eigenvector}: choose $k$ nodes  with highest eigenvector centrality associated with leading eigenvalue of adjacency matrix.
	\item \textsc{Betweenness}: choose $k$ nodes with highest betweenness.
	\item \textsc{Closeness}: choose $k$ nodes with highest closeness  centrality.
\end{enumerate}








\subsection{Effectiveness}

To demonstrate the   effectiveness of our algorithms \textsc{Na\"{\i}ve} and \textsc{Fast}, we first execute experiments  on four small networks with less than 100 nodes: Dolphins, Tribes, Karate and Firm-Hi-Tech,  comparing  our algorithms with \textsc{Optimal} and \textsc{Degree} schemes. Since computing the optimal solution requires  exponential time, we only consider five cases of $k=1,2,3,4,5$, and  display the results in  Fig.~\ref{fig:lambda}. One can observes from Fig.~\ref{fig:lambda} that our algorithms \textsc{Na\"{\i}ve} and \textsc{Fast} have little difference, with both being close to the optimum results and better than \textsc{Degree} scheme. Although Fig.~\ref{fig:lambda}(c) shows that \textsc{Degree} scheme may output the optimal solution for  Karate networks when $k=2$ as our algorithms, for $k=3,4,5$ our two approaches have better  effectiveness than the \textsc{Degree} scheme.

\begin{figure}[htbp]
  \centering
  \includegraphics[width=.95\linewidth]{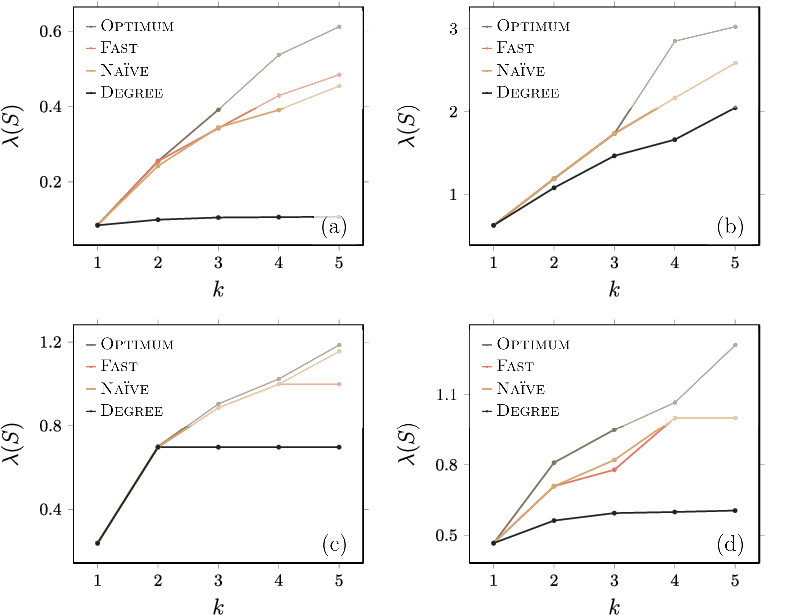}\\
\caption{$\lambda(S)$ given by \textsc{Fast}, \textsc{Na\"{\i}ve}, \textsc{Degree}, and \textsc{Optimum} schemes for $k$ ranging from 1 to 5 on four small networks: (a) Dolphins, (b) Tribes, (c) Karate and (d) Firm-Hi-Tech.}
\label{fig:lambda}
\end{figure}



In the above, we have shown that our algorithms \textsc{Na\"{\i}ve} and \textsc{Fast} exhibit similar effectiveness, which can be understood as follows. Since the difference between \textsc{Na\"{\i}ve} and \textsc{Fast} lies in the computation of $\lambda_S(j)=\lambda(S+j)-\lambda(S)$, the eigenvalue gap when a new node $j$ is added to set $S$. In \textsc{Na\"{\i}ve} , $\lambda_S(j)$ is directly computed by solving the eigensystem, while  in \textsc{Fast}, $\lambda_S(j)$ is evaluated by Eq.~\eqref{eq:c1}. Next, we show that the two methods for calculating $\lambda_S(j)$ in algorithms \textsc{Na\"{\i}ve} and $\tilde{\lambda}_S(j)$ in \textsc{Fast}  yield almost the same effect. To this end, we conduct  experiments on four networks: Dolphins, Tribes, Karate and Email-Univ. For each network,  we first randomly select five nodes forming  set $S$, and then add one more node into $S$ by computing the eigenvalue gap through the methods in \textsc{Na\"{\i}ve} and \textsc{Fast}. Figure~\ref{fig:lambda} reports the results for $\lambda_S(\cdot)$ and $\tilde{\lambda}_S(\cdot)$ for the two methods, which shows that the  eigenvalue gaps returned by the two approaches  are linearly proportional to each other,  justifying the similarity of \textsc{Na\"{\i}ve} and \textsc{Fast}.


\begin{figure}[htbp]
  \centering
  \includegraphics[width=1.0\linewidth]{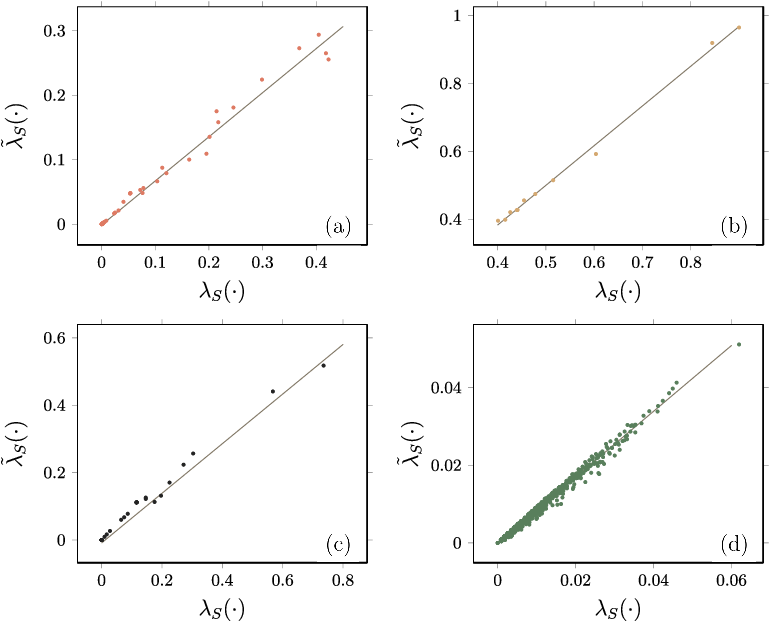}\\
\caption{Eigenvalue gap returned by \textsc{Na\"{\i}ve} denoted by $\lambda_S(\cdot)$ verse eigenvalue gap returned by \textsc{Fast} denoted by $\tilde{\lambda}_S(\cdot)$ on four  networks: (a) Dolphins, (b) Tribes, (c) Karate, and (d) Email-Univ.}
\label{fig:lambda}
\end{figure}

We further demonstrate the efficacy of  our algorithms \textsc{Fast} by comparing it  with four baseline  methods frequently used in previous work: \textsc{Degree},  \textsc{Eigenvector}, \textsc{Betweenness}, and \textsc{Closeness} on six networks, with the cardinality $k$ of set $S$ changing from $1$ to $100$. The comparison results are reported  in Fig.~\ref{fig:graph3}, from which we can observe that our proposed algorithm \textsc{Fast} almost outperforms the baseline methods.


\begin{figure}[htbp]
  \centering
  \includegraphics[width=1.0\linewidth]{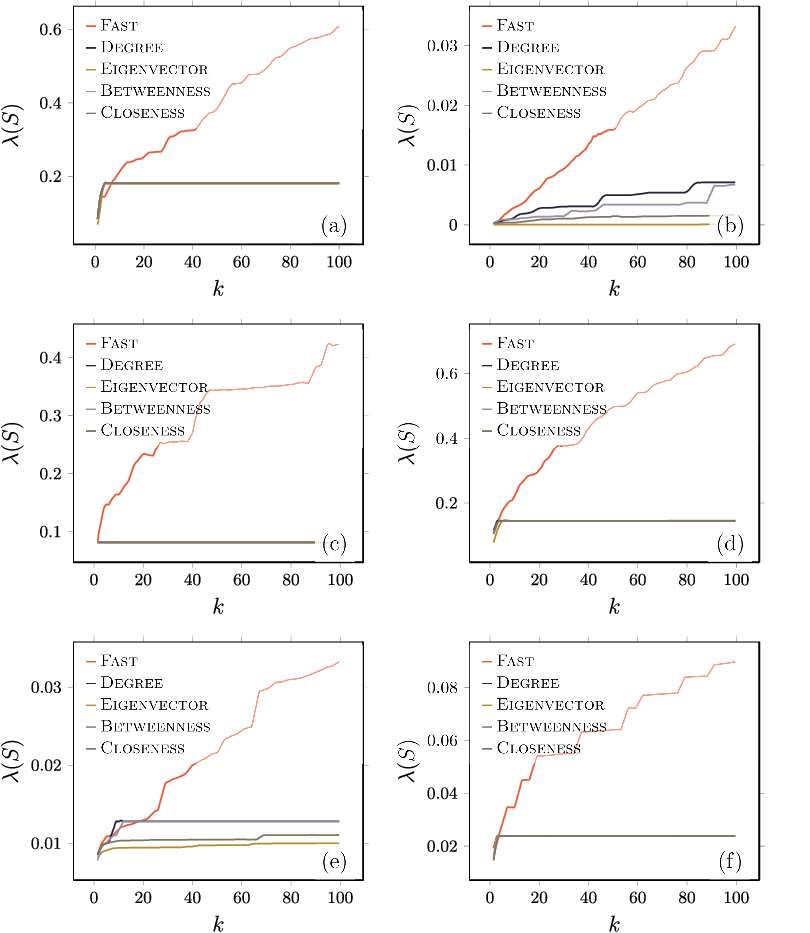}\\
\caption{$\lambda(S)$ returned  by \textsc{Fast},  \textsc{Degree},  \textsc{Eigenvector}, \textsc{Betweenness}, and \textsc{Closeness}, for $k$ ranging from 1 to 100 on six medium-size networks: (a) Pages-Government, (b) US-Grid, (c) Anybeat, (d) WHOIS, (e) Pretty Good Privacy, and (f) Epinions. }
\label{fig:graph3}
\end{figure}



\begin{figure}[htbp]
  \centering
  \includegraphics[width=.95\linewidth]{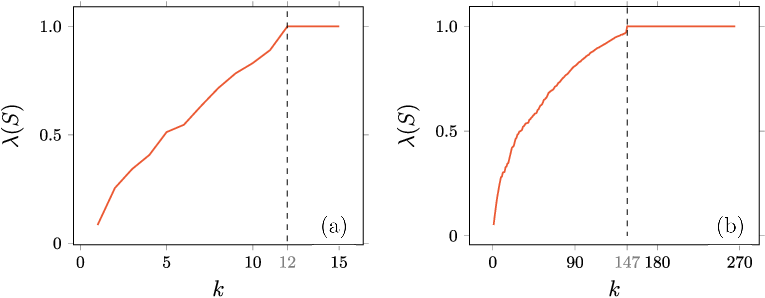}\\
\caption{$\lambda(S)$ returned  by \textsc{Fast}, on two networks: (a) Dolphins and (b) Email-Univ. The dashed line corresponds to the earliest moment when $\lambda(S)=1$ appears. }
\label{fig:graph5}
\end{figure}

To compare with the most dominant work~\cite{LiXuLuChZe21} on this optimization problem, we also repeat the experiment on Dolphins and Email-Univ in~\cite{LiXuLuChZe21}. According to~\cite{LiXuLuChZe21}, the Dolphins and Email-Univ's $\lambda(S)$ are doomed to be not smaller than 1 if $S$ are their dominating sets respectively. And it takes 14 and 266 nodes for method from~\cite{LiXuLuChZe21} to make $\lambda(S)$ arrive this bound. However, as shown in Fig.~\ref{fig:graph5}, our method only takes 12 and 147 nodes to grounded to let $\lambda(S)=1$. Besides, our method is not limited by $k$'s scale but method from~\cite{LiXuLuChZe21} cannot sustain oversized $k$.

\subsection{Efficiency and Scalability}

Although both \textsc{Fast} and \textsc{Na\"{\i}ve} are very accurate, the theoretical computation complexity of the two algorithms differs greatly. Here we experimentally compare the running time of \textsc{Fast} and \textsc{Exaxt} in various real networks, which is reported in Table~\ref{tab:time1}. The results show that algorithm \textsc{Fast} runs much faster than algorithm \textsc{Na\"{\i}ve}, especially on large-scale networks. For example, for the networks with more than 40,000 nodes marked with $*$, algorithm \textsc{Na\"{\i}ve} fails due to its tremendous time cost, while algorithm \textsc{Fast} takes only a few seconds for each iteration even on the network with over one million nodes, indicating the scalability of \textsc{Approx} to large networks.



\begin{table}[htbp]
		\centering
		\fontsize{9}{10}\selectfont
		\begin{threeparttable}
			\caption{Average running time in one iteration for computing $\lambda_S(\cdot)$  and $C_S(\cdot)$ defined in the caption of Figure~\ref{fig:lambda} and their ratio for all nodes on a larger set of real-world networks.}
			\label{tab:time1}
			\begin{tabular}{ccccc}
				\toprule
				\multirow{2}{*}{Network}&
\multirow{2}{*}{}&
				\multicolumn{3}{c}{Time (seconds)}\cr
				\cmidrule(lr){3-5}
				&& \textsc{Exact} & \textsc{Approx} &Ratio\cr
				\midrule
				
				685-bus&
				& 4.73  & 0.004 & 1182
				 \cr
				
				Bcspwr09&
				& 22.60  & 0.011  & 2054
				\cr
				
				Routers-RF&
				& 52.15  & 0.046  & 1133
				 \cr
				
				Bcspwr10&
				& 318.4  & 0.044  & 7236
				 \cr
				
				Webbase-2001&
				& 2613  & 0.134  & 19500
				 \cr
				
				As-Caida2007&
				& 5308   & 0.261  & 20337
				 \cr
				
				Epinions&
				& 4931  & 0.229  & 21532
				 \cr
				
				Email-EU&
				&  7148  & 0.126  & 56730   \cr
				
				Internet-as*&
				& -  & 0.227  & -  \cr
				
				P2P-gnutella*&
				& -  & 0.482  & -  \cr
				
				RL-caida*&
				& -  &1.992 & - \cr
				
				DBLP-2010*&
				& -  & 2.141  & -   \cr
				
				Twitter-follows*&
				& -  & 2.292 & -   \cr
				
				Delicious*&
				& -  & 4.731  & -   \cr

				FourSquare*&
				& -  &  16.33 & -   \cr

				Youtube-snap*&
				& -  &  31.72 & -  \cr
				
				\bottomrule
			\end{tabular}
		\end{threeparttable}
		\vspace{-5pt}
	\end{table}

We also experimentally compare the two techniques, Lanzcos method and SDDM solver, for computing the eigenvector $\uu$ associated with the smallest eigenvalue of the grounded Laplacian matrix, in terms of  the computation efficiency. For this purpose, we randomly select the grounded node set $S$ with size ranging from 1 to 10, and average the time for computing  eigenvector $\uu$, as reported in Table~\ref{tab:time2}. From the table, one can see that Lanzcos method is faster than SDDM solver when the networks relatively small. However, for larger networks with over 40,000 nodes, SDDM solver starts to show its advantages in terms of efficiency. For instance, for the last two networks with over 600,000 nodes, Lanzcos method fails while SDDM solver can still output the eigenvector in less than half a minute for each iteration. Thus, in all our experiments except those with results in Table~\ref{tab:time2}, the eigenvector corresponding to the smallest eigenvaluewe of  grounded Laplacian matrix is evaluated by SDDM solver, instead of Lanzcos method.


\begin{table}[htbp]
		\centering
		\fontsize{9}{9}\selectfont
		\begin{threeparttable}
			\caption{Average running time for Lanzcos method and Laplaican solver  for computing the
eigenvector  associated with the smallest eigenvalue of the grounded Laplacian matrix for some real-world networks.}
			\label{tab:time2}
			\begin{tabular}{ccc}
				\toprule
				\multirow{2}{*}{Network}&
				\multicolumn{2}{c}{Time (seconds)}\cr
				\cmidrule(lr){2-3}
				& Lanzcos Method & Inverse Power Method \cr
				\midrule

				685-bus&
				 0.002  &0.004
				 \cr
				
				Bcspwr09&
				 0.004  & 0.008
				\cr
				
				Routers-RF&
				 0.011  & 0.016
				 \cr
				
				Bcspwr10&
				 0.013  & 0.033
				 \cr
				
				Webbase-2001&
				 0.020  & 0.102
				 \cr
				
				As-Caida2007&
				 0.116   & 0.193
				 \cr
				
				Epinions&
				 0.441  & 0.411
				 \cr
				
				Email-EU&
				  0.205  & 0.157     \cr
				
				Internet-as&
				 0.199  & 0.211    \cr
				
				P2P-gnutella&
				 4.062  & 0.379   \cr
				
				RL-caida&
			  3.081 &2.034 \cr
				
				DBLP-2010&
				 8.390  & 1.157    \cr
				
				Twitter-follows&
				 3.328  & 1.371    \cr
				
				Delicious&
				 39.90  & 4.961    \cr

				FourSquare&
				 -  &  8.035   \cr

				Youtube-snap&
				 -  &  29.07 \cr

				\bottomrule
			\end{tabular}
		\end{threeparttable}
		\vspace{-5pt}
	\end{table}


In fact, in addition to  \textsc{Fast}, the two baseline methods  \textsc{Degree} and  \textsc{Eigenvector} are also scalable.  As shown in Fig.~\ref{fig:graph3}, \textsc{Degree} displays the best effect among the all heuristic baseline methods, in the following, we will compare  \textsc{Fast} with only \textsc{Degree}, the complexity of which is $O(n)$. Specifically, we compare \textsc{Fast} with  \textsc{Degree}  on four real networks with size ranging from 200,000 to 1,000,000,  and report comparison results  in Fig.~\ref{fig:bigg}.  From  Fig.~\ref{fig:bigg}, one observes that  \textsc{Fast} outperforms  \textsc{Degree}  on all the four studied networks.




\begin{figure}[htbp]
  \centering
  \includegraphics[width=.85\linewidth]{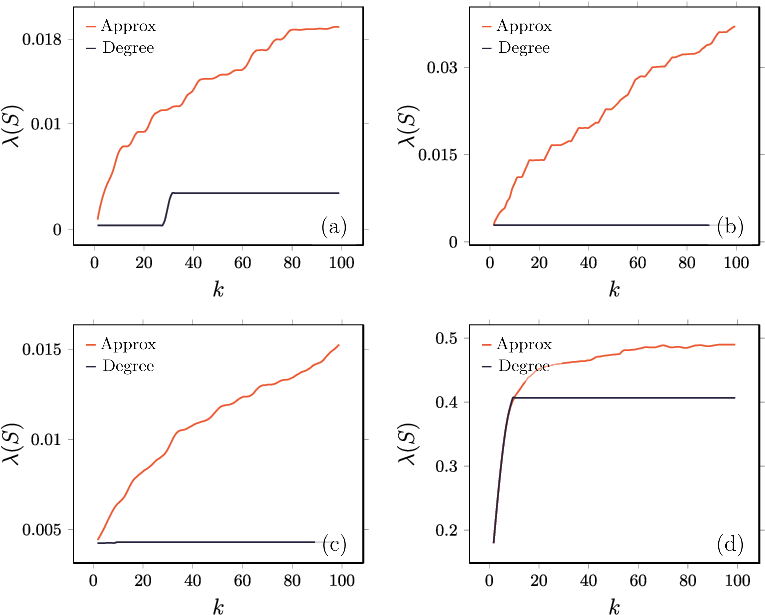}\\
\caption{$\lambda(S)$ returned  by \textsc{Fast} with  \textsc{Degree} with $k$ ranging from 1 to 100 on four large  networks: (a) RL-Caida, (b) Delicious, (c) Youtube-Snap, (d) FourSquare. }
\label{fig:bigg}
\end{figure}

\section{Conclusions}

In this paper, we studied the minimum eigenvalue $\lambda(S)$ for a grounded Laplacian matrix $\LL(S)$ of order $(n-k) \times (n-k)$, which is a principal submatrix of the Laplacian matrix $\LL$ for a graph with $n$ nodes and $m$ edges, obtained from $\LL$ by removing $k$ rows and columns corresponding to $k$ selected nodes forming set $S$. We focused on the problem of finding the set $S^*$ of $k$ nodes with an aim to maximize $\lambda(S^*)$. This problem arises naturally in various contexts, such as maximizing the convergence speed of leader-follower consensus dynamics and maximizing the effectiveness of pinning scheme of pinning control problem. We provided a rigorous proof of the NP-hardness of this combination optimization problem, and proved that the objective function is not submodular, despite its weak strict monotonicity. We developed an approximation algorithm with time complexity $\Otil (mk)$, which maximizes $\lambda(S)$ by iteratively choosing $k$ nodes in a greedy way. Extensive experiments on various realistic networks show our algorithm is both effective and efficient, which gives almost optimal solutions and is scalable to huge networks with more than $10^6$ nodes.

\section*{Declarations}
\bmhead{Availability of data and material}
Not applicable.

\bmhead{Conflict of interest}
There are no conflict of interest or competing interests related to this manuscript.

\bmhead{Acknowledgments}
This work was supported by the National Natural Science Foundation of China (Nos. 61872093 and U20B2051),  Shanghai Municipal Science and Technology Major Project  (Nos.  2018SHZDZX01 and 2021SHZDZX03),  ZJ Lab, and Shanghai Center for Brain Science and Brain-Inspired Technology. Run Wang was also supported by Fudan’s Undergraduate Research Opportunities Program (FDUROP) under Grant No. 2195200241021.

\bibliographystyle{sn-mathphys}

\bibliography{reference,EigenOpt}


\begin{thebibliography}{68}
\ifx \bisbn   \undefined \def \bisbn  #1{ISBN #1}\fi
\ifx \binits  \undefined \def \binits#1{#1}\fi
\ifx \bauthor  \undefined \def \bauthor#1{#1}\fi
\ifx \batitle  \undefined \def \batitle#1{#1}\fi
\ifx \bjtitle  \undefined \def \bjtitle#1{#1}\fi
\ifx \bvolume  \undefined \def \bvolume#1{\textbf{#1}}\fi
\ifx \byear  \undefined \def \byear#1{#1}\fi
\ifx \bissue  \undefined \def \bissue#1{#1}\fi
\ifx \bfpage  \undefined \def \bfpage#1{#1}\fi
\ifx \blpage  \undefined \def \blpage #1{#1}\fi
\ifx \burl  \undefined \def \burl#1{\textsf{#1}}\fi
\ifx \doiurl  \undefined \def \doiurl#1{\url{https://doi.org/#1}}\fi
\ifx \betal  \undefined \def \betal{\textit{et al.}}\fi
\ifx \binstitute  \undefined \def \binstitute#1{#1}\fi
\ifx \binstitutionaled  \undefined \def \binstitutionaled#1{#1}\fi
\ifx \bctitle  \undefined \def \bctitle#1{#1}\fi
\ifx \beditor  \undefined \def \beditor#1{#1}\fi
\ifx \bpublisher  \undefined \def \bpublisher#1{#1}\fi
\ifx \bbtitle  \undefined \def \bbtitle#1{#1}\fi
\ifx \bedition  \undefined \def \bedition#1{#1}\fi
\ifx \bseriesno  \undefined \def \bseriesno#1{#1}\fi
\ifx \blocation  \undefined \def \blocation#1{#1}\fi
\ifx \bsertitle  \undefined \def \bsertitle#1{#1}\fi
\ifx \bsnm \undefined \def \bsnm#1{#1}\fi
\ifx \bsuffix \undefined \def \bsuffix#1{#1}\fi
\ifx \bparticle \undefined \def \bparticle#1{#1}\fi
\ifx \barticle \undefined \def \barticle#1{#1}\fi
\bibcommenthead
\ifx \bconfdate \undefined \def \bconfdate #1{#1}\fi
\ifx \botherref \undefined \def \botherref #1{#1}\fi
\ifx \url \undefined \def \url#1{\textsf{#1}}\fi
\ifx \bchapter \undefined \def \bchapter#1{#1}\fi
\ifx \bbook \undefined \def \bbook#1{#1}\fi
\ifx \bcomment \undefined \def \bcomment#1{#1}\fi
\ifx \oauthor \undefined \def \oauthor#1{#1}\fi
\ifx \citeauthoryear \undefined \def \citeauthoryear#1{#1}\fi
\ifx \endbibitem  \undefined \def \endbibitem {}\fi
\ifx \bconflocation  \undefined \def \bconflocation#1{#1}\fi
\ifx \arxivurl  \undefined \def \arxivurl#1{\textsf{#1}}\fi
\csname PreBibitemsHook\endcsname

\bibitem{Ne03}
\begin{barticle}
\bauthor{\bsnm{Newman}, \binits{M.E.J.}}:
\batitle{The structure and function of complex networks}.
\bjtitle{SIAM Rev.}
\bvolume{45}(\bissue{2}),
\bfpage{167}--\blpage{256}
(\byear{2003})
\end{barticle}
\endbibitem

\bibitem{WaChWaFa03}
\begin{bchapter}
\bauthor{\bsnm{Wang}, \binits{Y.}},
\bauthor{\bsnm{Chakrabarti}, \binits{D.}},
\bauthor{\bsnm{Wang}, \binits{C.}},
\bauthor{\bsnm{Faloutsos}, \binits{C.}}:
\bctitle{Epidemic spreading in real networks: An eigenvalue viewpoint}.
In: \bbtitle{Proc. 22nd Int. Symp. Reliab. Distrib. Syst.},
pp. \bfpage{25}--\blpage{34}
(\byear{2003}).
\bcomment{IEEE}
\end{bchapter}
\endbibitem

\bibitem{ChWaWaLeFa08}
\begin{barticle}
\bauthor{\bsnm{Chakrabarti}, \binits{D.}},
\bauthor{\bsnm{Wang}, \binits{Y.}},
\bauthor{\bsnm{Wang}, \binits{C.}},
\bauthor{\bsnm{Leskovec}, \binits{J.}},
\bauthor{\bsnm{Faloutsos}, \binits{C.}}:
\batitle{Epidemic thresholds in real networks}.
\bjtitle{ACM Trans. Inf. Sys. Secur.}
\bvolume{10}(\bissue{4}),
\bfpage{1}--\blpage{26}
(\byear{2008})
\end{barticle}
\endbibitem

\bibitem{VaOmKo08}
\begin{barticle}
\bauthor{\bsnm{Van~Mieghem}, \binits{P.}},
\bauthor{\bsnm{Omic}, \binits{J.}},
\bauthor{\bsnm{Kooij}, \binits{R.}}:
\batitle{Virus spread in networks}.
\bjtitle{IEEE/ACM Trans. Netw.}
\bvolume{17}(\bissue{1}),
\bfpage{1}--\blpage{14}
(\byear{2008})
\end{barticle}
\endbibitem

\bibitem{BoBoChRi10}
\begin{barticle}
\bauthor{\bsnm{Bollob{\'a}s}, \binits{B.}},
\bauthor{\bsnm{Borgs}, \binits{C.}},
\bauthor{\bsnm{Chayes}, \binits{J.}},
\bauthor{\bsnm{Riordan}, \binits{O.}}:
\batitle{Percolation on dense graph sequences}.
\bjtitle{Ann. Prob.}
\bvolume{38}(\bissue{1}),
\bfpage{150}--\blpage{183}
(\byear{2010})
\end{barticle}
\endbibitem

\bibitem{OlMu04}
\begin{barticle}
\bauthor{\bsnm{Olfati-Saber}, \binits{R.}},
\bauthor{\bsnm{Murray}, \binits{R.M.}}:
\batitle{Consensus problems in networks of agents with switching topology and
  time-delays}.
\bjtitle{IEEE Trans. Autom. Control}
\bvolume{49}(\bissue{9}),
\bfpage{1520}--\blpage{1533}
(\byear{2004})
\end{barticle}
\endbibitem

\bibitem{LiPaYiZh20}
\begin{barticle}
\bauthor{\bsnm{Li}, \binits{H.}},
\bauthor{\bsnm{Patterson}, \binits{S.}},
\bauthor{\bsnm{Yi}, \binits{Y.}},
\bauthor{\bsnm{Zhang}, \binits{Z.}}:
\batitle{Maximizing the number of spanning trees in a connected graph}.
\bjtitle{IEEE Trans. Inf. Theory}
\bvolume{66}(\bissue{2}),
\bfpage{1248}--\blpage{1260}
(\byear{2020})
\end{barticle}
\endbibitem

\bibitem{KlRa93}
\begin{barticle}
\bauthor{\bsnm{Klein}, \binits{D.J.}},
\bauthor{\bsnm{Randi{\'c}}, \binits{M.}}:
\batitle{Resistance distance}.
\bjtitle{J. Math. Chem.}
\bvolume{12}(\bissue{1}),
\bfpage{81}--\blpage{95}
(\byear{1993})
\end{barticle}
\endbibitem

\bibitem{LiZh18}
\begin{bchapter}
\bauthor{\bsnm{Li}, \binits{H.}},
\bauthor{\bsnm{Zhang}, \binits{Z.}}:
\bctitle{Kirchhoff index as a measure of edge centrality in weighted networks:
  {N}early linear time algorithms}.
In: \bbtitle{Proc. 29th Ann. ACM-SIAM Symp. Disc. Alg.},
pp. \bfpage{2377}--\blpage{2396}
(\byear{2018})
\end{bchapter}
\endbibitem

\bibitem{Te91}
\begin{barticle}
\bauthor{\bsnm{Tetali}, \binits{P.}}:
\batitle{Random walks and the effective resistance of networks}.
\bjtitle{J. Theor. Probab.}
\bvolume{4}(\bissue{1}),
\bfpage{101}--\blpage{109}
(\byear{1991})
\end{barticle}
\endbibitem

\bibitem{ChRaRuSm89}
\begin{bchapter}
\bauthor{\bsnm{Chandra}, \binits{A.K.}},
\bauthor{\bsnm{Raghavan}, \binits{P.}},
\bauthor{\bsnm{Ruzzo}, \binits{W.L.}},
\bauthor{\bsnm{Smolensky}, \binits{R.}}:
\bctitle{The electrical resistance of a graph captures its commute and cover
  times}.
In: \bbtitle{Proc. 21st {ACM} Symp. Theory Comput.},
pp. \bfpage{574}--\blpage{586}
(\byear{1989})
\end{bchapter}
\endbibitem

\bibitem{ShZh19}
\begin{barticle}
\bauthor{\bsnm{Sheng}, \binits{Y.}},
\bauthor{\bsnm{Zhang}, \binits{Z.}}:
\batitle{Low-mean hitting time for random walks on heterogeneous networks}.
\bjtitle{IEEE Trans. Inf. Theory}
\bvolume{65}(\bissue{11}),
\bfpage{6898}--\blpage{6910}
(\byear{2019})
\end{barticle}
\endbibitem

\bibitem{BaJoMiPa12}
\begin{barticle}
\bauthor{\bsnm{Bamieh}, \binits{B.}},
\bauthor{\bsnm{Jovanovic}, \binits{M.R.}},
\bauthor{\bsnm{Mitra}, \binits{P.}},
\bauthor{\bsnm{Patterson}, \binits{S.}}:
\batitle{Coherence in large-scale networks: Dimension-dependent limitations of
  local feedback}.
\bjtitle{IEEE Trans. Autom. Control}
\bvolume{57},
\bfpage{2235}--\blpage{2249}
(\byear{2012})
\end{barticle}
\endbibitem

\bibitem{VeBoPa15}
\begin{barticle}
\bauthor{\bsnm{Veremyev}, \binits{A.}},
\bauthor{\bsnm{Boginski}, \binits{V.}},
\bauthor{\bsnm{Pasiliao}, \binits{E.L.}}:
\batitle{Analytical characterizations of some classes of optimal strongly
  attack-tolerant networks and their laplacian spectra}.
\bjtitle{Journal of Global Optimization}
\bvolume{61}(\bissue{1}),
\bfpage{109}--\blpage{138}
(\byear{2015})
\end{barticle}
\endbibitem

\bibitem{QiZhYiLi19}
\begin{barticle}
\bauthor{\bsnm{Qi}, \binits{Y.}},
\bauthor{\bsnm{Zhang}, \binits{Z.}},
\bauthor{\bsnm{Yi}, \binits{Y.}},
\bauthor{\bsnm{Li}, \binits{H.}}:
\batitle{Consensus in self-similar hierarchical graphs and {S}ierpi{\'n}ski
  graphs: {C}onvergence speed, delay robustness, and coherence}.
\bjtitle{IEEE Trans. Cybern.}
\bvolume{49}(\bissue{2}),
\bfpage{592}--\blpage{603}
(\byear{2019})
\end{barticle}
\endbibitem

\bibitem{YiZhPa20}
\begin{barticle}
\bauthor{\bsnm{Yi}, \binits{Y.}},
\bauthor{\bsnm{Zhang}, \binits{Z.}},
\bauthor{\bsnm{Patterson}, \binits{S.}}:
\batitle{Scale-free loopy structure is resistant to noise in consensus dynamics
  in complex networks}.
\bjtitle{IEEE Trans. Cybern.}
\bvolume{50}(\bissue{1}),
\bfpage{190}--\blpage{200}
(\byear{2020})
\end{barticle}
\endbibitem

\bibitem{KaVePa17}
\begin{barticle}
\bauthor{\bsnm{Kammerdiner}, \binits{A.}},
\bauthor{\bsnm{Veremyev}, \binits{A.}},
\bauthor{\bsnm{Pasiliao}, \binits{E.}}:
\batitle{On laplacian spectra of parametric families of closely connected
  networks with application to cooperative control}.
\bjtitle{Journal of Global Optimization}
\bvolume{67},
\bfpage{187}--\blpage{205}
(\byear{2017})
\end{barticle}
\endbibitem

\bibitem{ChZhZhYuLi21}
\begin{barticle}
\bauthor{\bsnm{Chen}, \binits{X.}},
\bauthor{\bsnm{Zhang}, \binits{S.}},
\bauthor{\bsnm{Zhang}, \binits{L.}},
\bauthor{\bsnm{Yu}, \binits{G.}},
\bauthor{\bsnm{Liu}, \binits{J.}}:
\batitle{Determining redundant links of multiagent systems in keeping or
  improving consensus convergence rates}.
\bjtitle{IEEE Systems Journal}
\bvolume{16}(\bissue{4}),
\bfpage{6153}--\blpage{6163}
(\byear{2021})
\end{barticle}
\endbibitem

\bibitem{ChGaZhZh21}
\begin{barticle}
\bauthor{\bsnm{Chen}, \binits{X.}},
\bauthor{\bsnm{Gao}, \binits{S.}},
\bauthor{\bsnm{Zhang}, \binits{S.}},
\bauthor{\bsnm{Zhao}, \binits{Y.}}:
\batitle{On topology optimization for event-triggered consensus with triggered
  events reducing and convergence rate improving}.
\bjtitle{IEEE Transactions on Circuits and Systems II: Express Briefs}
\bvolume{69}(\bissue{3}),
\bfpage{1223}--\blpage{1227}
(\byear{2021})
\end{barticle}
\endbibitem

\bibitem{GoPeBiVo21}
\begin{barticle}
\bauthor{\bsnm{Gorbunov}, \binits{A.}},
\bauthor{\bsnm{Peng}, \binits{J.C.-H.}},
\bauthor{\bsnm{Bialek}, \binits{J.W.}},
\bauthor{\bsnm{Vorobev}, \binits{P.}}:
\batitle{Identification of stability regions in inverter-based microgrids}.
\bjtitle{IEEE Transactions on Power Systems}
\bvolume{37}(\bissue{4}),
\bfpage{2613}--\blpage{2623}
(\byear{2021})
\end{barticle}
\endbibitem

\bibitem{ZhZhLuLi23}
\begin{barticle}
\bauthor{\bsnm{Zhang}, \binits{Y.}},
\bauthor{\bsnm{Zhou}, \binits{J.}},
\bauthor{\bsnm{Lu}, \binits{J.-a.}},
\bauthor{\bsnm{Li}, \binits{W.}}:
\batitle{Superdiffusion induced by complete structure in multiplex networks}.
\bjtitle{Chaos: An Interdisciplinary Journal of Nonlinear Science}
\bvolume{33}(\bissue{2}),
\bfpage{023133}
(\byear{2023})
\end{barticle}
\endbibitem

\bibitem{RaJiMeEg09}
\begin{barticle}
\bauthor{\bsnm{Rahmani}, \binits{A.}},
\bauthor{\bsnm{Ji}, \binits{M.}},
\bauthor{\bsnm{Mesbahi}, \binits{M.}},
\bauthor{\bsnm{Egerstedt}, \binits{M.}}:
\batitle{Controllability of multi-agent systems from a graph-theoretic
  perspective}.
\bjtitle{SIAM J. Control Optimiz.}
\bvolume{48}(\bissue{1}),
\bfpage{162}--\blpage{186}
(\byear{2009})
\end{barticle}
\endbibitem

\bibitem{PaBa10}
\begin{bchapter}
\bauthor{\bsnm{Patterson}, \binits{S.}},
\bauthor{\bsnm{Bamieh}, \binits{B.}}:
\bctitle{Leader selection for optimal network coherence}.
In: \bbtitle{Proc. 49th {IEEE} Conf. Decision Control},
pp. \bfpage{2692}--\blpage{2697}
(\byear{2010}).
\bcomment{IEEE}
\end{bchapter}
\endbibitem

\bibitem{PiShFiSu18}
\begin{barticle}
\bauthor{\bsnm{Pirani}, \binits{M.}},
\bauthor{\bsnm{Shahrivar}, \binits{E.M.}},
\bauthor{\bsnm{Fidan}, \binits{B.}},
\bauthor{\bsnm{Sundaram}, \binits{S.}}:
\batitle{Robustness of leader-follower networked dynamical systems}.
\bjtitle{IEEE Trans. Control Netw. Syst.}
\bvolume{5}(\bissue{4}),
\bfpage{1752}--\blpage{1763}
(\byear{2018})
\end{barticle}
\endbibitem

\bibitem{LiXuLuChZe21}
\begin{barticle}
\bauthor{\bsnm{Liu}, \binits{H.}},
\bauthor{\bsnm{Xu}, \binits{X.}},
\bauthor{\bsnm{Lu}, \binits{J.}},
\bauthor{\bsnm{Chen}, \binits{G.}},
\bauthor{\bsnm{Zeng}, \binits{Z.}}:
\batitle{Optimizing pinning control of complex dynamical networks based on
  spectral properties of grounded laplacian matrices}.
\bjtitle{IEEE Trans. Syst. Man Cybern. -Syst.}
\bvolume{51}(\bissue{2}),
\bfpage{786}--\blpage{796}
(\byear{2021})
\end{barticle}
\endbibitem

\bibitem{Mi93}
\begin{barticle}
\bauthor{\bsnm{Miekkala}, \binits{U.}}:
\batitle{Graph properties for splitting with grounded {L}aplacian matrices}.
\bjtitle{Bit}
\bvolume{33}(\bissue{3}),
\bfpage{485}--\blpage{495}
(\byear{1993})
\end{barticle}
\endbibitem

\bibitem{BaHe06}
\begin{bchapter}
\bauthor{\bsnm{Barooah}, \binits{P.}},
\bauthor{\bsnm{Hespanha}, \binits{J.P.}}:
\bctitle{Graph effective resistance and distributed control: {S}pectral
  properties and applications}.
In: \bbtitle{Proc. 45th IEEE Conf. Decision Control},
pp. \bfpage{3479}--\blpage{3485}
(\byear{2006}).
\bcomment{IEEE}
\end{bchapter}
\endbibitem

\bibitem{WaCh02}
\begin{barticle}
\bauthor{\bsnm{Wang}, \binits{X.F.}},
\bauthor{\bsnm{Chen}, \binits{G.}}:
\batitle{Pinning control of scale-free dynamical networks}.
\bjtitle{Physica A}
\bvolume{310}(\bissue{3-4}),
\bfpage{521}--\blpage{531}
(\byear{2002})
\end{barticle}
\endbibitem

\bibitem{LiWaCh04}
\begin{barticle}
\bauthor{\bsnm{Li}, \binits{X.}},
\bauthor{\bsnm{Wang}, \binits{X.}},
\bauthor{\bsnm{Chen}, \binits{G.}}:
\batitle{Pinning a complex dynamical network to its equilibrium}.
\bjtitle{IEEE Trans. Circuits Syst. I-Reg Papers}
\bvolume{51}(\bissue{10}),
\bfpage{2074}--\blpage{2087}
(\byear{2004})
\end{barticle}
\endbibitem

\bibitem{HeMaHuSe15}
\begin{barticle}
\bauthor{\bsnm{Herman}, \binits{I.}},
\bauthor{\bsnm{Martinec}, \binits{D.}},
\bauthor{\bsnm{Hur{\'a}k}, \binits{Z.}},
\bauthor{\bsnm{{\v{S}}ebek}, \binits{M.}}:
\batitle{Nonzero bound on fiedler eigenvalue causes exponential growth of
  h-infinity norm of vehicular platoon}.
\bjtitle{IEEE Trans. Autom. Control}
\bvolume{60}(\bissue{8}),
\bfpage{2248}--\blpage{2253}
(\byear{2015})
\end{barticle}
\endbibitem

\bibitem{TeBaGa15}
\begin{barticle}
\bauthor{\bsnm{Tegling}, \binits{E.}},
\bauthor{\bsnm{Bamieh}, \binits{B.}},
\bauthor{\bsnm{Gayme}, \binits{D.F.}}:
\batitle{The price of synchrony: Evaluating the resistive losses in
  synchronizing power networks}.
\bjtitle{IEEE Trans. Control Netw. Syst.}
\bvolume{2}(\bissue{3}),
\bfpage{254}--\blpage{266}
(\byear{2015})
\end{barticle}
\endbibitem

\bibitem{PiSu14}
\begin{bchapter}
\bauthor{\bsnm{Pirani}, \binits{M.}},
\bauthor{\bsnm{Sundaram}, \binits{S.}}:
\bctitle{Spectral properties of the grounded {L}aplacian matrix with
  applications to consensus in the presence of stubborn agents}.
In: \bbtitle{2014 Amer. Control Conf.},
pp. \bfpage{2160}--\blpage{2165}
(\byear{2014}).
\bcomment{IEEE}
\end{bchapter}
\endbibitem

\bibitem{PiSu16}
\begin{barticle}
\bauthor{\bsnm{Pirani}, \binits{M.}},
\bauthor{\bsnm{Sundaram}, \binits{S.}}:
\batitle{On the smallest eigenvalue of grounded {L}aplacian matrices}.
\bjtitle{IEEE Trans. Autom. Control}
\bvolume{61}(\bissue{2}),
\bfpage{509}--\blpage{514}
(\byear{2016})
\end{barticle}
\endbibitem

\bibitem{PiShSu15}
\begin{bchapter}
\bauthor{\bsnm{Pirani}, \binits{M.}},
\bauthor{\bsnm{Shahrivar}, \binits{E.M.}},
\bauthor{\bsnm{Sundaram}, \binits{S.}}:
\bctitle{Coherence and convergence rate in networked dynamical systems}.
In: \bbtitle{Proc. 54th IEEE Conf. Decision Control},
pp. \bfpage{968}--\blpage{973}
(\byear{2015}).
\bcomment{IEEE}
\end{bchapter}
\endbibitem

\bibitem{MaBe17}
\begin{barticle}
\bauthor{\bsnm{Manaffam}, \binits{S.}},
\bauthor{\bsnm{Behal}, \binits{A.}}:
\batitle{Bounds on the smallest eigenvalue of a pinned {L}aplacian matrix}.
\bjtitle{IEEE Trans. Autom. Control}
\bvolume{63}(\bissue{8}),
\bfpage{2641}--\blpage{2646}
(\byear{2017})
\end{barticle}
\endbibitem

\bibitem{XiCa17}
\begin{barticle}
\bauthor{\bsnm{Xia}, \binits{W.}},
\bauthor{\bsnm{Cao}, \binits{M.}}:
\batitle{Analysis and applications of spectral properties of grounded
  {L}aplacian matrices for directed networks}.
\bjtitle{Automatica}
\bvolume{80},
\bfpage{10}--\blpage{16}
(\byear{2017})
\end{barticle}
\endbibitem

\bibitem{VaAnMe20}
\begin{barticle}
\bauthor{\bparticle{van~der} \bsnm{Grinten}, \binits{A.}},
\bauthor{\bsnm{Angriman}, \binits{E.}},
\bauthor{\bsnm{Meyerhenke}, \binits{H.}}:
\batitle{Scaling up network centrality computations--{A} brief overview}.
\bjtitle{it-Inform. Technol.}
\bvolume{62}(\bissue{3-4}),
\bfpage{189}--\blpage{204}
(\byear{2020})
\end{barticle}
\endbibitem

\bibitem{DoelPuZi09}
\begin{barticle}
\bauthor{\bsnm{Dolev}, \binits{S.}},
\bauthor{\bsnm{Elovici}, \binits{Y.}},
\bauthor{\bsnm{Puzis}, \binits{R.}},
\bauthor{\bsnm{Zilberman}, \binits{P.}}:
\batitle{Incremental deployment of network monitors based on group betweenness
  centrality}.
\bjtitle{Inform. Proces. Lett.}
\bvolume{109}(\bissue{20}),
\bfpage{1172}--\blpage{1176}
(\byear{2009})
\end{barticle}
\endbibitem

\bibitem{MaTsUp16}
\begin{bchapter}
\bauthor{\bsnm{Mahmoody}, \binits{A.}},
\bauthor{\bsnm{Tsourakakis}, \binits{C.E.}},
\bauthor{\bsnm{Upfal}, \binits{E.}}:
\bctitle{Scalable betweenness centrality maximization via sampling}.
In: \bbtitle{Proc. 22nd ACM SIGKDD Int. Conf. Knowl. Discovery Data Mining},
pp. \bfpage{1765}--\blpage{1773}
(\byear{2016}).
\bcomment{ACM}
\end{bchapter}
\endbibitem

\bibitem{EvBo99}
\begin{barticle}
\bauthor{\bsnm{Everett}, \binits{M.G.}},
\bauthor{\bsnm{Borgatti}, \binits{S.P.}}:
\batitle{The centrality of groups and classes}.
\bjtitle{J. Math. Sociol.}
\bvolume{23}(\bissue{3}),
\bfpage{181}--\blpage{201}
(\byear{1999})
\end{barticle}
\endbibitem

\bibitem{BeGoMe18}
\begin{bchapter}
\bauthor{\bsnm{Bergamini}, \binits{E.}},
\bauthor{\bsnm{Gonser}, \binits{T.}},
\bauthor{\bsnm{Meyerhenke}, \binits{H.}}:
\bctitle{Scaling up group closeness maximization}.
In: \bbtitle{Proc. 12th Workshop Algorithm Engin. Experim.},
pp. \bfpage{209}--\blpage{222}
(\byear{2018}).
\bcomment{SIAM}
\end{bchapter}
\endbibitem

\bibitem{LiPeShYiZh19}
\begin{bchapter}
\bauthor{\bsnm{Li}, \binits{H.}},
\bauthor{\bsnm{Peng}, \binits{R.}},
\bauthor{\bsnm{Shan}, \binits{L.}},
\bauthor{\bsnm{Yi}, \binits{Y.}},
\bauthor{\bsnm{Zhang}, \binits{Z.}}:
\bctitle{Current flow group closeness centrality for complex networks?}
In: \bbtitle{Proc. World Wide Web Conf.},
pp. \bfpage{961}--\blpage{971}
(\byear{2019})
\end{bchapter}
\endbibitem

\bibitem{GhTeLeYa14}
\begin{bchapter}
\bauthor{\bsnm{Ghosh}, \binits{R.}},
\bauthor{\bsnm{Teng}, \binits{S.-h.}},
\bauthor{\bsnm{Lerman}, \binits{K.}},
\bauthor{\bsnm{Yan}, \binits{X.}}:
\bctitle{The interplay between dynamics and networks: centrality, communities,
  and cheeger inequality}.
In: \bbtitle{Proc. 20th ACM SIGKDD Int. Conf. Knowl. Discovery Data Mining},
pp. \bfpage{1406}--\blpage{1415}
(\byear{2014}).
\bcomment{ACM}
\end{bchapter}
\endbibitem

\bibitem{WaLiXuLi19}
\begin{bchapter}
\bauthor{\bsnm{Wang}, \binits{B.}},
\bauthor{\bsnm{Liu}, \binits{H.}},
\bauthor{\bsnm{Xu}, \binits{J.}},
\bauthor{\bsnm{Liu}, \binits{J.}}:
\bctitle{Pining control algorithm for complex networks}.
In: \bbtitle{Proc. 2019 Chin. Control Conf.},
pp. \bfpage{964}--\blpage{969}
(\byear{2019}).
\bcomment{IEEE}
\end{bchapter}
\endbibitem

\bibitem{ClAlBuPo12}
\begin{bchapter}
\bauthor{\bsnm{Clark}, \binits{A.}},
\bauthor{\bsnm{Alomair}, \binits{B.}},
\bauthor{\bsnm{Bushnell}, \binits{L.}},
\bauthor{\bsnm{Poovendran}, \binits{R.}}:
\bctitle{Leader selection in multi-agent systems for smooth convergence via
  fast mixing}.
In: \bbtitle{Proc. 51st IEEE Conf. Decision Control},
pp. \bfpage{818}--\blpage{824}
(\byear{2012}).
\bcomment{IEEE}
\end{bchapter}
\endbibitem

\bibitem{ClBuPo12}
\begin{bchapter}
\bauthor{\bsnm{Clark}, \binits{A.}},
\bauthor{\bsnm{Bushnell}, \binits{L.}},
\bauthor{\bsnm{Poovendran}, \binits{R.}}:
\bctitle{Leader selection for minimizing convergence error in leader-follower
  systems: A supermodular optimization approach}.
In: \bbtitle{Proc. 10th Int. Symp. Model. Optimiz. Mobile, Ad Hoc and Wirel.
  Netw.},
pp. \bfpage{111}--\blpage{115}
(\byear{2012}).
\bcomment{IEEE}
\end{bchapter}
\endbibitem

\bibitem{ClHoBuPo18}
\begin{barticle}
\bauthor{\bsnm{Clark}, \binits{A.}},
\bauthor{\bsnm{Hou}, \binits{Q.}},
\bauthor{\bsnm{Bushnell}, \binits{L.}},
\bauthor{\bsnm{Poovendran}, \binits{R.}}:
\batitle{Maximizing the smallest eigenvalue of a symmetric matrix: A submodular
  optimization approach}.
\bjtitle{Automatica}
\bvolume{95},
\bfpage{446}--\blpage{454}
(\byear{2018})
\end{barticle}
\endbibitem

\bibitem{ZhTa20}
\begin{bchapter}
\bauthor{\bsnm{Zhou}, \binits{J.}},
\bauthor{\bsnm{Tang}, \binits{W.K.}}:
\bctitle{Feature-embedded evolutionary algorithm for network optimization}.
In: \bbtitle{Proc. 2020 IEEE Int. Symp. Circuits Syst.},
pp. \bfpage{1}--\blpage{5}
(\byear{2020}).
\bcomment{IEEE}
\end{bchapter}
\endbibitem

\bibitem{McNeScTs95}
\begin{barticle}
\bauthor{\bsnm{McDonald}, \binits{J.J.}},
\bauthor{\bsnm{Neumann}, \binits{M.}},
\bauthor{\bsnm{Schneider}, \binits{H.}},
\bauthor{\bsnm{Tsatsomeros}, \binits{M.J.}}:
\batitle{Inverse {M}-matrix inequalities and generalized ultrametric matrices}.
\bjtitle{Linear Algebra Appl.}
\bvolume{220},
\bfpage{321}--\blpage{341}
(\byear{1995})
\end{barticle}
\endbibitem

\bibitem{Ma00}
\begin{barticle}
\bauthor{\bsnm{MacCluer}, \binits{C.R.}}:
\batitle{{The many proofs and applications of Perron's theorem}}.
\bjtitle{SIAM Rev.}
\bvolume{42}(\bissue{3}),
\bfpage{487}--\blpage{498}
(\byear{2000})
\end{barticle}
\endbibitem

\bibitem{NeWoFi78}
\begin{barticle}
\bauthor{\bsnm{Nemhauser}, \binits{G.L.}},
\bauthor{\bsnm{Wolsey}, \binits{L.A.}},
\bauthor{\bsnm{Fisher}, \binits{M.L.}}:
\batitle{{An analysis of approximations for maximizing submodular set
  functions-I}}.
\bjtitle{Math. Program.}
\bvolume{14}(\bissue{1}),
\bfpage{265}--\blpage{294}
(\byear{1978})
\end{barticle}
\endbibitem

\bibitem{FrHeJa98}
\begin{barticle}
\bauthor{\bsnm{Fricke}, \binits{G.}},
\bauthor{\bsnm{Hedetniemi}, \binits{S.T.}},
\bauthor{\bsnm{Jacobs}, \binits{D.P.}}:
\batitle{Independence and irredundance in $k$-regular graphs}.
\bjtitle{Ars Comb.}
\bvolume{49},
\bfpage{271}--\blpage{279}
(\byear{1998})
\end{barticle}
\endbibitem

\bibitem{Ba10}
\begin{bbook}
\bauthor{\bsnm{Bapat}, \binits{R.B.}}:
\bbtitle{Graphs and Matrices}.
\bpublisher{New York: Springer}, \blocation{???}
(\byear{2010})
\end{bbook}
\endbibitem

\bibitem{La52}
\begin{botherref}
\oauthor{\bsnm{Lanczos}, \binits{C.}}:
Solution of systems of linear equations by minimized iterations1.
J. Res. Natl. Inst. Bereau Stand.
\textbf{49}(1)
(1952)
\end{botherref}
\endbibitem

\bibitem{Ch82}
\begin{barticle}
\bauthor{\bsnm{Chaiken}, \binits{S.}}:
\batitle{A combinatorial proof of the all minors matrix tree theorem}.
\bjtitle{SIAM J. Alg. Discrete Met.}
\bvolume{3}(\bissue{3}),
\bfpage{319}--\blpage{329}
(\byear{1982})
\end{barticle}
\endbibitem

\bibitem{YiSh2018}
\begin{bchapter}
\bauthor{\bsnm{Yi}, \binits{Y.}},
\bauthor{\bsnm{Shan}, \binits{L.}},
\bauthor{\bsnm{Li}, \binits{H.}},
\bauthor{\bsnm{Zhang}, \binits{Z.}}:
\bctitle{Biharmonic distance related centrality for edges in weighted
  networks.}
In: \bbtitle{Proc. 27th Int. Joint Conf. Art. Intel.},
pp. \bfpage{3620}--\blpage{3626}
(\byear{2018})
\end{bchapter}
\endbibitem

\bibitem{SiBoBaMo18}
\begin{barticle}
\bauthor{\bsnm{Siami}, \binits{M.}},
\bauthor{\bsnm{Bolouki}, \binits{S.}},
\bauthor{\bsnm{Bamieh}, \binits{B.}},
\bauthor{\bsnm{Motee}, \binits{N.}}:
\batitle{Centrality measures in linear consensus networks with structured
  network uncertainties}.
\bjtitle{IEEE Trans. Control Netw. Syst.}
\bvolume{5}(\bissue{3}),
\bfpage{924}--\blpage{934}
(\byear{2018})
\end{barticle}
\endbibitem

\bibitem{KaTo19}
\begin{bchapter}
\bauthor{\bsnm{Kang}, \binits{J.}},
\bauthor{\bsnm{Tong}, \binits{H.}}:
\bctitle{{N2N:} network derivative mining}.
In: \bbtitle{Proceedings of the 28th {ACM} Int. Conf. Inf. Knowledge Manage.,
  {CIKM} 2019, Beijing, China, November 3-7, 2019},
pp. \bfpage{861}--\blpage{870}.
\bpublisher{{ACM}}, \blocation{???}
(\byear{2019})
\end{bchapter}
\endbibitem

\bibitem{Mi11}
\begin{bbook}
\bauthor{\bsnm{Mieghem}, \binits{P.}}:
\bbtitle{Graph Spectra for Complex Networks}.
\bpublisher{Cambridge University Press}, \blocation{???}
(\byear{2011})
\end{bbook}
\endbibitem

\bibitem{HeYaYuZh19}
\begin{botherref}
\oauthor{\bsnm{He}, \binits{Z.}},
\oauthor{\bsnm{Yao}, \binits{C.}},
\oauthor{\bsnm{Yu}, \binits{J.}},
\oauthor{\bsnm{Zhan}, \binits{M.}}:
Perturbation analysis and comparison of network synchronization methods.
Phys. Rev. E
\textbf{99}
(2019)
\end{botherref}
\endbibitem

\bibitem{MiSuNi10}
\begin{botherref}
\oauthor{\bsnm{Milanese}, \binits{A.}},
\oauthor{\bsnm{Sun}, \binits{J.}},
\oauthor{\bsnm{Nishikawa}, \binits{T.}}:
Approximating spectral impact of structural perturbations in large networks.
Phys. Rev. E
\textbf{81}
(2010)
\end{botherref}
\endbibitem

\bibitem{ReOtHu06}
\begin{botherref}
\oauthor{\bsnm{Restrepo}, \binits{J.G.}},
\oauthor{\bsnm{Ott}, \binits{E.}},
\oauthor{\bsnm{Hunt}, \binits{B.R.}}:
Characterizing the dynamical importance of network nodes and links.
Phys. Rev. Lett.
\textbf{97}
(2006)
\end{botherref}
\endbibitem

\bibitem{BoPaRa99}
\begin{barticle}
\bauthor{\bsnm{Bocea}, \binits{M.}},
\bauthor{\bsnm{Panagiotopoulos}, \binits{P.D.}},
\bauthor{\bsnm{R{\u{a}}dulescu}, \binits{V.}}:
\batitle{A perturbation result for a double eigenvalue hemivariational
  inequality with constraints and applications}.
\bjtitle{Journal of Global Optimization}
\bvolume{14},
\bfpage{137}--\blpage{156}
(\byear{1999})
\end{barticle}
\endbibitem

\bibitem{BaSpSrTe13}
\begin{barticle}
\bauthor{\bsnm{Batson}, \binits{J.}},
\bauthor{\bsnm{Spielman}, \binits{D.A.}},
\bauthor{\bsnm{Srivastava}, \binits{N.}},
\bauthor{\bsnm{Teng}, \binits{S.H.}}:
\batitle{Spectral sparsification of graphs: {T}heory and algorithms}.
\bjtitle{Commun. ACM}
\bvolume{56}(\bissue{8}),
\bfpage{87}--\blpage{94}
(\byear{2013})
\end{barticle}
\endbibitem

\bibitem{SpTe14}
\begin{barticle}
\bauthor{\bsnm{Spielman}, \binits{D.A.}},
\bauthor{\bsnm{Teng}, \binits{S.-H.}}:
\batitle{Nearly linear time algorithms for preconditioning and solving
  symmetric, diagonally dominant linear systems}.
\bjtitle{SIAM J. Matrix Anal. Appl.}
\bvolume{35}(\bissue{3}),
\bfpage{835}--\blpage{885}
(\byear{2014})
\end{barticle}
\endbibitem

\bibitem{CoKyMiPaJaPeRaXu14}
\begin{bchapter}
\bauthor{\bsnm{Cohen}, \binits{M.B.}},
\bauthor{\bsnm{Kyng}, \binits{R.}},
\bauthor{\bsnm{Miller}, \binits{G.L.}},
\bauthor{\bsnm{Pachocki}, \binits{J.W.}},
\bauthor{\bsnm{Peng}, \binits{R.}},
\bauthor{\bsnm{Rao}, \binits{A.B.}},
\bauthor{\bsnm{Xu}, \binits{S.C.}}:
\bctitle{Solving {SDD} linear systems in nearly $m \log^{1/2} n$ time}.
In: \bbtitle{Proc. 46th Annu. ACM Symp. Theory Comput.},
pp. \bfpage{343}--\blpage{352}
(\byear{2014}).
\bcomment{ACM}
\end{bchapter}
\endbibitem

\bibitem{Ku13}
\begin{bchapter}
\bauthor{\bsnm{Kunegis}, \binits{J.}}:
\bctitle{Konect: the koblenz network collection}.
In: \bbtitle{Proc. 22nd Int. Conf. World Wide Web},
pp. \bfpage{1343}--\blpage{1350}
(\byear{2013}).
\bcomment{ACM}
\end{bchapter}
\endbibitem

\bibitem{LeSo16}
\begin{barticle}
\bauthor{\bsnm{Leskovec}, \binits{J.}},
\bauthor{\bsnm{Sosi{\v{c}}}, \binits{R.}}:
\batitle{{SNAP}: A general-purpose network analysis and graph-mining library}.
\bjtitle{ACM Trans. Intel. Syst. Tech.}
\bvolume{8}(\bissue{1}),
\bfpage{1}
(\byear{2016})
\end{barticle}
\endbibitem

\end{thebibliography}

%

\end{document}